\documentclass[12pt]{article}
\pdfoutput=1
\usepackage{graphicx}
\usepackage{amssymb}
\usepackage{amsthm}
\usepackage{amsmath}
\usepackage[left=1.2in,right=1.2in,top=1.2in,bottom=1.2in]{geometry}
\usepackage{mathtools}
\usepackage{caption}
\usepackage[caption=false]{subfig}
\newtheorem{lemma}{Lemma}
\newtheorem{theorem}{Theorem}

\newcommand{\g}{\mathfrak{g}}
\newcommand{\rme}{\mathrm{e}}
\newcommand{\rmd}{\mathrm{d}}

\begin{document}

\begin{center}
\vspace{24pt}
{ \large \bf Trees and spatial topology change in CDT}

\vspace{30pt}

{\sl J. Ambj\o rn}$\,^{a,b}$
and {\sl T.G. Budd}$\,^{a}$

\vspace{48pt}
{\footnotesize

$^a$~The Niels Bohr Institute, Copenhagen University\\
Blegdamsvej 17, DK-2100 Copenhagen \O , Denmark.\\
{ email: ambjorn@nbi.dk, budd@nbi.dk}\\

\vspace{10pt}
$^b$~Institute for Mathematics, Astrophysics and Particle Physics (IMAPP)\\ 
Radbaud University Nijmegen, Heyendaalseweg 135,
6525 AJ, Nijmegen, The Netherlands 
}
\vspace{36pt}
\end{center}


\begin{center}
{\bf Abstract}
\end{center}
\noindent
Generalized causal dynamical triangulations (generalized CDT) is a model of two-dimensional quantum gravity in which a limited number of spatial topology changes is allowed to occur.
We solve the model at the discretized level using bijections between quadrangulations and trees.
In the continuum limit (scaling limit) the amplitudes are shown to agree with known formulas and explicit expressions are obtained for loop propagators and two-point functions.
It is shown that from a combinatorial point of view generalized CDT can be viewed as the scaling limit of planar maps with a finite number of faces and we determine the distance function on this ensemble 
of planar maps.  
Finally, the relation with planar maps is used to illuminate a mysterious identity of certain continuum cylinder amplitudes.

\vspace{12pt}
\noindent

\vspace{24pt}

\newpage

\section{Introduction}\label{sec:intro}

Two-dimensional quantum gravity has been an important topic in  theoretical 
physics for a long time.  String theory {\it is} two-dimensional 
quantum gravity coupled to certain conformal field theories in its
simplest perturbative formulation. When the central charge $c$ of a 
conformal field theory is less than one the models correspond to 
non-critical string theories, and it is possible to solve 
certain aspects of the gravity-matter system analytically.
Also, for $c< 1$ it is possible to provide a path integral 
regularization of these quantum theories. In this regularization
one performs the integration over 2d geometries by summing over 
equilateral triangulations (so-called \emph{dynamical triangulations} (DT)), 
eventually recovering the continuum limit by taking the length $\epsilon$
of the links to zero. Remarkably, a class of these regularized theories 
can be solved analytically, even for $\epsilon > 0$,  using 
combinatorial techniques, either by directly counting certain graphs or 
by using so-called matrix models. The outcome of this 
has been a beautiful Wilsonian picture where one has universality:
the continuum limit is to a large extent independent of the details
of the regularization. It does not really matter if the starting point is 
triangulations or one uses  quadrangulations or higher order 
polygons  in an arbitrary  combination,
as long as the weights of these are positive (see e.g.\ \cite{ambjorn_quantum_1997}, chapter
4, for a review). Thus one has an 
infinite dimensional coupling constant space, the coupling constants
being the relative weights of various types of polygons, and the 
critical surface where the continuum limit can be taken is a 
hyper-surface of finite co-dimension. On this hyper-surface one 
obtains ``pure'' 2d Euclidean quantum gravity. If one allows negative
weights for some polygons one can flow to new continuum theories
describing 2d Euclidean quantum gravity coupled to various 
conformal matter theories, and if one allows for various ``flavors''
to be attached to the polygons and a local interaction
between these, one can obtain all mininal, rational conformal 
field theories coupled to 2d Euclidean quantum gravity in the 
continuum limit. 

These 2d discretized models can be generalized to higher dimensions 
\cite{ambjorn_four-dimensional_1992,agishtein_simulations_1992}.
However it has so far not been possible to find a \emph{continuum limit}
which can be viewed as higher dimensional quantum gravity. This 
failure triggered an attempt to define a new class of regularized 
models where the sums over the piecewise linear geometries were
first carried out in spacetimes with Lorentzian signature and 
local causality was imposed. This class of piecewise linear
geometries  was denoted \emph{causal dynamical triangulations} (CDT) 
\cite{ambjorn_non-perturbative_1998,ambjorn_dynamically_2001}. 
When rotating back to Euclidean signature, one is effectively 
summing over triangulations where there is a \emph{proper-time foliation}.
It seems that this class of models has an interesting continuum limit
both in three and four dimensions (for a review see \cite{ambjorn_nonperturbative_2012}), 
but until now this has only been 
investigated using computer simulations (in three dimensions somewhat 
related models have been looked at analytically). Here we will concentrate
on 2d where the regularized model can be solved analytically using 
combinatorial methods. More specifically we will consider a 
model known as \emph{generalized CDT}, which interpolates between the CDT 
and DT models \cite{ambjorn_putting_2007,ambjorn_string_2008}. 
In the generalized CDT model one starts out with 
space being connected, i.e. having the topology of $S^1$,  and
as a function of proper time one allows it to split into 
a finite number of $S^1$ components.
Using recent combinatorial results we will  
solve the discretized  model and show how one 
can obtain the scaling or continuum limit of the model.
Further we will discuss  how this limit relates to the standard 
CDT and DT limits. 
To simplify the discussion we use a model which at the discretized
level is described by quadrangulations, rather than triangulations.
(in section \ref{sec:triangulations} we show how the results can be generalized
to triangulations). 

The rest of the article is organized as follows:
in section \ref{sec:bijections} we review how quadrangulations of the sphere
are related to labeled planar trees in the case of DT and to unlabeled planar trees
in the case of CDT. In section \ref{sec:gencdt} we show how one can use these trees to study a discrete version of generalized CDT and how to find its scaling limit.
A more detailed counting of labeled trees in section \ref{sec:timedep} allows us to study proper time depences in generalized CDT. 
In section \ref{sec:planarmaps} it is shown that generalized CDT can also be interpreted as a scaling limit of random planar maps for which the number of faces is conditioned to remain finite.

There exists an intriguing identity  
between certain cylinder amplitudes in the scaling limit, first discussed 
in the context of DT by Kawai and Ishibashi \cite{ishibashi_string_1994} and later 
in the context of generalized CDT in \cite{ambjorn_string_2008}. Kawai and 
Ishabashi related the identities to a Virasoro algebra and an 
underlying conformal invariance, but in the context of  
generalized CDT the identities appeared quite mysterious. However,
we will show that they are even valid at the discretized 
level and can be understood as a bijection between sets of 
quadrangulations defining the cylinder amplitudes in question.
This is discussed in section \ref{sec:loops}.

Finally, in section \ref{sec:triangulations} it is shown how some of the quadrangulation results can be generalized to triangulations.

\section{Bijections}\label{sec:bijections}
\subsection{Definitions}\label{sec:definitions}

In the following we will make extensive use of planar trees, quadrangulations, triangulations and more general tilings of the 2-sphere.
In order to facilitate the discussion it is useful to recognize these objects as subclasses of the more general notion of planar maps.

An \emph{embedded planar graph} is a multigraph, i.e. a graph in which edges are allowed to begin and end at the same vertex and in which multiple edges are allowed between pairs of vertices, embedded in the 2-sphere without crossing edges.
A \emph{planar map} is a connected embedded planar graph. 
Two planar maps are considered equivalent if they can be continuously deformed into each other. 
An equivalence class of planar maps can be described purely combinatorically, e.g. by labeling the vertices and writing down for each vertex an (clockwise) ordered list of its neighbors in the graph.  
The connected components of the complement of an embedded planar graph in the sphere are called the \emph{faces} of the graph.
In the case of a planar map all faces are topological disks and the number of edges bounding a face is called the \emph{degree} of the face.
A face of degree $d$ also has $d$ corners, where by a \emph{corner} we mean a small sector around a vertex bounded by two consecutive edges of the face.
Notice that a single vertex may appear more than once as a corner of a face.

When it comes to counting it is often convenient to deal with objects with no internal symmetries.
In the case of planar maps, which may have a non-trivial automorphism group, the symmetry can be explicitly broken by \emph{rooting} the map, which means that one marks an oriented edge, the \emph{root edge}.
The face directly to the left of the root edge is called the \emph{root face} and its corner associated with the end of the root edge is the \emph{root corner}.
Clearly selecting a root edge is equivalent to selecting a root corner, and therefore either type of root may be used interchangeably.  

Various classes of planar maps will be of interest to us. 
First of all, \emph{triangulations} and \emph{quadrangulations}, respectively, are planar maps for which each face has degree 3 and 4.
More generally, a \emph{$p$-angulation}, $p\geq 2$ is a planar map for which each face has degree $p$.
A \emph{$p$-angulation with $b$ boundaries of length $l_i$}, $i=1,\ldots,b$, is a planar map with $b$ distinguished faces with degree $l_i$ and all other faces having degree $p$.
Finally, a \emph{planar tree} is a planar map with only one face. 
The anti-clockwise ordered list of corners of the single face of a planar tree, starting at the root corner in the case of a \emph{rooted planar tree}, is often called the \emph{contour} of the tree.
It corresponds to the (periodic) list of vertices one encounters when walking around the tree in a \emph{clockwise} direction. 

In the following we will often study \emph{pointed quadrangulations}, which are quadrangulations with a marked vertex, called the \emph{origin}.
It is convenient to adopt a slightly different way of rooting a pointed quadrangulation, namely a \emph{rooted pointed quadrangulation} is a pointed quadrangulation with a marked (unoriented) edge.
As will become clear later, the marked vertex already induces a natural orientation on the edges and therefore a choice of orientation of the root edge would lead to a two-fold redundancy.

\subsection{Cori--Vauquelin--Schaeffer bijection}\label{sec:schaeffer}

The Cori--Vauquelin--Schaeffer bijection relates quadrangulations of the sphere with a marked vertex to planar trees with a labelling. Let us briefly describe the map and its inverse. For details and proofs we refer the reader to \cite{schaeffer_conjugaison_1998,chassaing_random_2004,gall_scaling_2011}, or to section \ref{sec:planarmaps} where we discuss and prove a more general bijection.

\begin{figure}[t]
{\centering
\includegraphics[width=\linewidth]{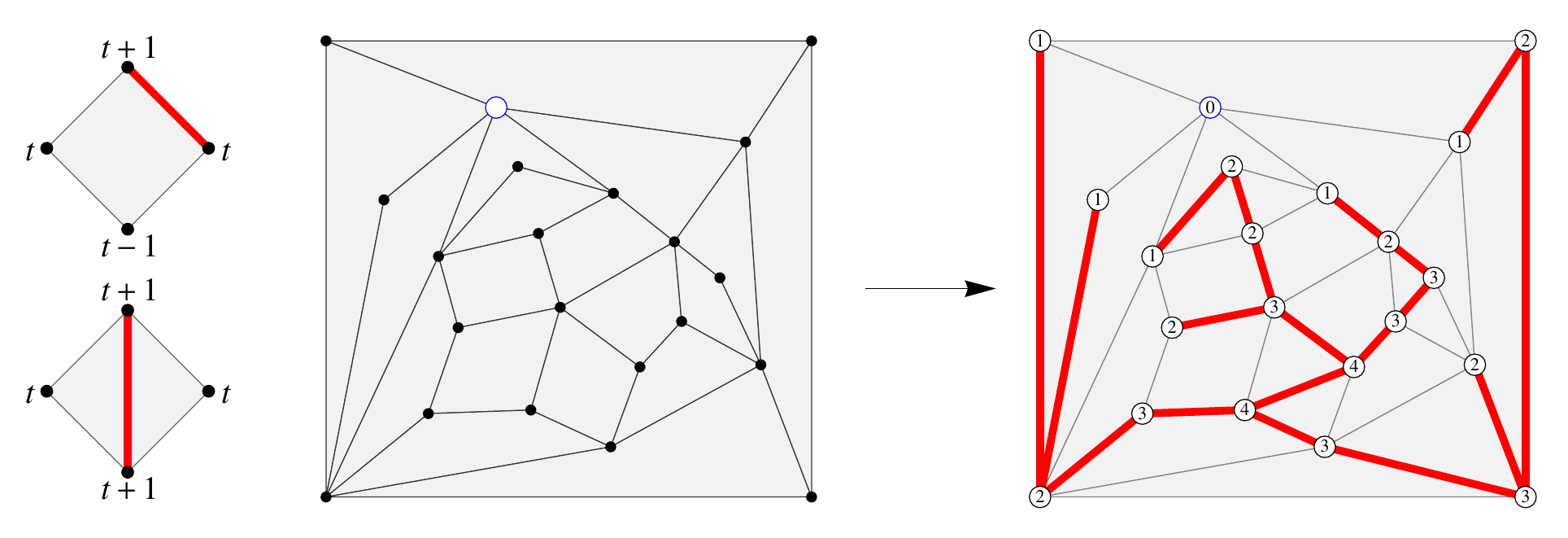}}
\vspace{-0.6cm}

\hspace{1cm}(a)\hspace{3.9cm}(b)\hspace{6.5cm}(c)
\caption{To every face of the quadrangulation one associates a coloured edge depending on the labeling. The coloured edges combine into a tree connecting all vertices except the origin.}%
\label{fig:quad0}
\end{figure}

Given a pointed quadrangulation of the sphere with $N$ faces, we can label the vertices by their distance to the origin along the edges of the quadrangulation.
The set of vertices naturally partitions into those with even respectively odd labels.
As a consequence of the faces having an even number of sides, this partition turns the edge graph into a bipartite graph, i.e. vertices with even label only connect to vertices with odd label and vice versa. 
Hence, each edge connects vertices with labels differing by exactly one.
Taking into account the labeling, two types of faces occur in the quadrangulation, \emph{simple} faces with labels $(t-1,t,t+1,t)$ and \emph{confluent} faces with labels $(t+1,t,t+1,t)$ in cyclic order (see figure \ref{fig:quad0}a and \cite{schaeffer_conjugaison_1998,chassaing_random_2004}).
A graph is drawn on the sphere by colouring the diagonal of each confluent face and the side of each simple face according to the prescription shown in figure \ref{fig:quad0}a. 
The resulting graph is a tree (with $N$ edges) containing all the vertices of the quadrangulation except for the origin. 
This is a consequence of the following lemma, which we formulate more generally than necessary here, since we will be reusing this result in section \ref{sec:planarmaps}.
\begin{lemma}\label{thm:qtopm}
Let $Q$ be a quadrangulation with $N$ faces and integer labels on its vertices, such that the labels differ by exactly one along the edges. 
Then the embedded planar graph $\mathcal{G}$ resulting from the prescription in figure \ref{fig:quad0}a is connected and therefore a planar map.
Moreover, each face of $\mathcal{G}$ contains in its interior exactly one vertex of $Q$.
These vertices are exactly the local minima of the labeling, i.e. vertices whose labels are smaller or equal to all the labels of their neighbors.
\end{lemma}
\begin{proof}
First we show that a vertex of $Q$ is a vertex of $\mathcal{G}$ if and only if it is not a local minimum.
Since the colouring in figure \ref{fig:quad0}a stays away from the minimal labels in the quadrangles, it is clear that a local minimum is not a vertex of $\mathcal{G}$. 
Given a vertex $v$ of $Q$ labeled $t$ that is not a local minimum, one can find an edge running from $v$ to a vertex labeled $t-1$. 
By inspection of figure \ref{fig:quad0}a we see that the quadrangle directly to the right of this edge in any case gives rise to a coloured edge ending at $v$.

Let us now prove that each face $\mathcal{F}$ of the embedded planar graph contains a vertex that is a local minimum.
For a given face $\mathcal{F}$, select from its corners one which has smallest possible label $t$.
There must exist an edge leading away from this vertex and having $\mathcal{F}$ on its left-hand side. This edge ends at a vertex that is labeled either $t$ or $t+1$. 
In the former case the edge is the diagonal of a confluent quadrangle, while in the latter case it is adjacent to a simple quadrangle on the left.
As can be seen from figure \ref{fig:quad0}a, in either case the quadrangle contains a vertex with label $t-1$ lying to the left of the edge.
Since this vertex cannot be a corner of $\mathcal{F}$ it must lie in the interior of $\mathcal{F}$.
Hence, there must be at least one local minimum in the interior of $\mathcal{F}$.

Let $V$, $E$ and $F$ be the number of vertices, edges and faces of $\mathcal{G}$, respectively.
The number of vertices of $Q$ is $N+2$, which follows from the Euler's formula, and each vertex of $Q$ that is not a local minimum belongs to $\mathcal{G}$, therefore $V=N-N_{\mathrm{min}}+2$.
By construction $E=N$ and from the considerations above it follows that $F \leq N_{\mathrm{min}}$.
Therefore we have $V-E+F = F-N_{\mathrm{min}} +2 \leq 2$.
However, according to the Euler characteristic for embedded planar graphs we have $V-E+F=1+C$, where $C\geq 1$ is the number of connected components of $\mathcal{G}$.
The only solution is $C=1$ and $F = N_{\min}$.
We conclude that $\mathcal{G}$ is connected and that each face of $\mathcal{G}$ contains exactly one local minimum in its interior.
\end{proof}

Since the distance labeling has only one local minimum, namely the origin itself, application of the lemma shows that the coloured graph is a planar map with a single face, hence a planar tree.
Keeping the labels on the vertices one ends up with a \emph{well-labeled tree}, i.e. a planar tree with positive integers on its vertices, such that at least one vertex is labeled 1 and the labels differ by at most one along its edges.

The quadrangulation can be reconstructed from the well-labeled tree in the following way. 
First we add by hand a new vertex in the plane, which will be the origin.
Then we consider the \emph{contour}, as defined above, containing the $2N$ corners of the tree in clockwise order. 
For every corner in the contour we draw a new edge in the plane: if the corner is labeled 1 we connect it to the origin; otherwise we connect it to the first corner following it that has smaller label.
Up to deformations there is a unique way of drawing all these edges without crossings.
After deleting the tree we are left with the original quadrangulation embedded in the plane.
For a proof of these statements see theorem \ref{thm:genbijection} in section \ref{sec:planarmaps}.

The bijection can be extended to the rooted versions of the quadrangulations and trees described above. 
Recall that a rooted pointed quadrangulation is a pointed quadrangulation where one of the edges is marked (indicated by a double-sided arrow in figure \ref{fig:quad1}a).
Since the edges of the quadrangulation are in 1-to-1 correspondence with the corners in the contour of the tree, we obtain a distinguished corner which we take to be the \emph{root} of the tree (indicated by a dark arrow in figure \ref{fig:quad1}b).
In a rooted tree all edges have a natural orientation, i.e. pointing away from the root.
To each edge of the tree we associate a label $+$, $0$ or $-$, depending on whether the label increases, remains the same, or decreases along the edge (see figure \ref{fig:quad1}c).
These labels are sufficient to reconstruct the labels on the vertices, since by construction the minimal label on the vertices is fixed to be equal to one.
From these constructions it follows that rooted pointed quadrangulations with $N$ faces are in bijection with rooted planar trees with $N$ edges labeled by $+$,$0$,$-$'s (see \cite{chassaing_random_2004}, Theorem 4). 

\begin{figure}[t]
{\centering
\includegraphics[width=\linewidth]{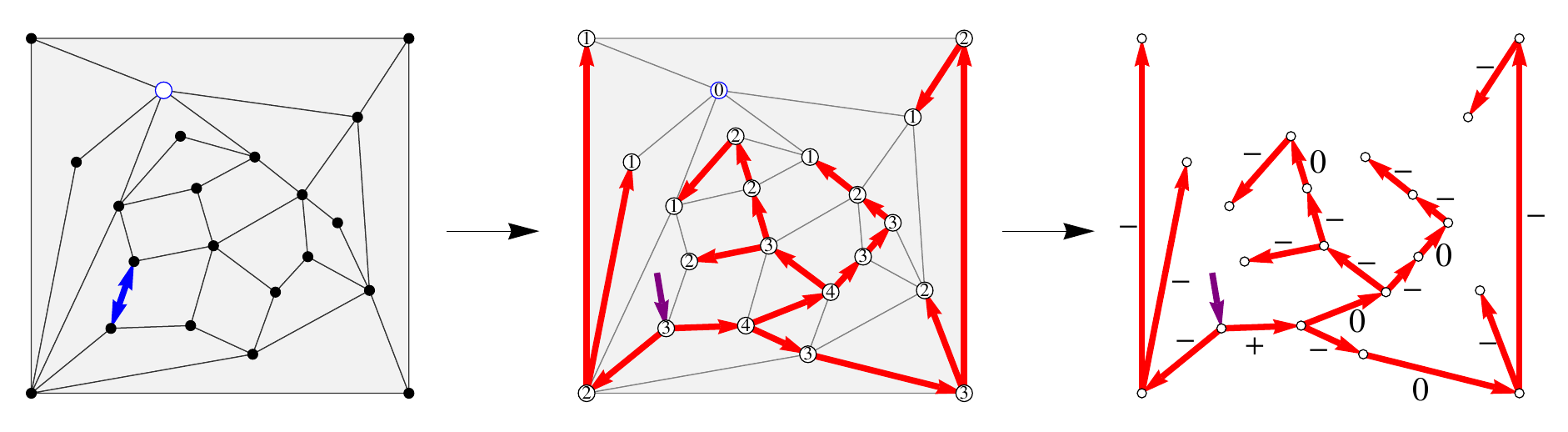}}
\vspace{-0.6cm}

\hspace{1.9cm}(a)\hspace{5.0cm}(b)\hspace{5.0cm}(c)
\caption{Rooted quadrangulations with an origin are in bijection with rooted planar trees labelled by $+$,$0$,$-$'s.}%
\label{fig:quad1}
\end{figure}

Notice that this bijection makes the counting of quadrangulations extremely simple. 
Rooted planar trees with $N$ edges are counted by the Catalan numbers 
\begin{equation}
C(N)=\frac{1}{N+1}\binom{2N}{N},
\end{equation}
while the number of labelings is simply $3^N$. 
The generating function $z^{(\ell)}(g)$ for the number of well-labeled rooted planar trees, and also for the number of rooted pointed quadrangulations, is therefore given by
\begin{equation}\label{eq:labeledgenfun}
z^{(\ell)}(g) = \sum_{N=0}^{\infty} 3^N C(N) g^N = \frac{1-\sqrt{1-12g}}{6g}.
\end{equation}
Since a quadrangulation of the sphere with $N$ faces has $2N$ edges and $N+2$ vertices, the micro-canonical partition function for unmarked quadrangulations is
\begin{equation}
Z(N)=\sum_{Q} \frac{1}{C_Q} = \frac{3^N}{2N(N+2)} C(N) = 
\frac{1}{2\sqrt{\pi}}N^{-7/2} 12^N (1+ \mathcal{O}(N^{-1})),
\end{equation}
where $C_Q$ is the order of the automorphism group of the quadrangulation $Q$.

\subsection{Causal triangulations}\label{sec:causal}

A similar bijection between causal triangulations and trees has been used in \cite{malyshev_two-dimensional_2001,krikun_phase_2012,durhuus_spectral_2010} and earlier in a slightly different form in \cite{di_francesco_integrable_2000}.
In analogy with the quadrangulations in the previous section we can define causal triangulations, which were introduced in \cite{ambjorn_non-perturbative_1998}, in the following way.
Consider a triangulation of the sphere with a marked vertex (the black point in figure \ref{fig:cdt0}), which we call the origin. 
We interpret the labeling of the vertices that arises from the distance to the origin as a time function, the CDT time.
The edges of the triangulation come in two types: \emph{spacelike} edges connecting vertices with identical labels and \emph{timelike} edges connecting vertices with different labels (dashed resp. solid edges in figure \ref{fig:cdt0}).
A triangulation with origin is a \emph{causal triangulation} when the graph consisting of only the spacelike edges is a disjoint union of cycles and there is exactly one vertex with maximal label.
In other words, the spatial topology as function of CDT time is fixed to be $S^1$.

\begin{figure}[t]
\centering
\subfloat[]{\centering\includegraphics[width=2.8cm]{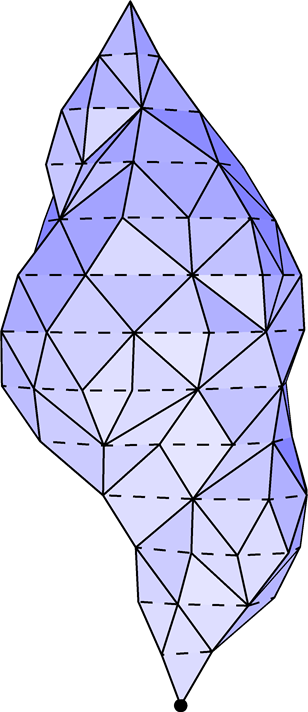}\label{fig:cdt0}}\hspace{0.3cm}
\subfloat[]{\centering\includegraphics[width=2.8cm]{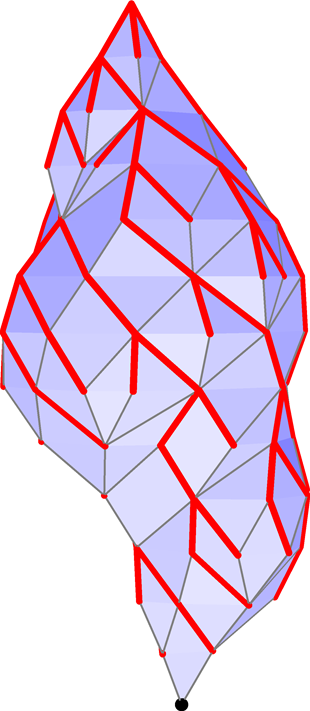}\label{fig:cdt1}}\hspace{0.3cm}
\subfloat[]{\centering\includegraphics[width=5.0cm]{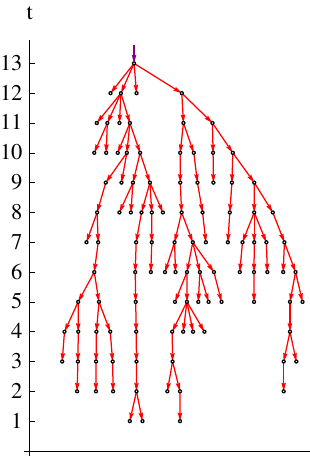}\label{fig:cdttree}}%
\caption{From (a) causal triangulations to (b) causal quadrangulations to (c) unlabeled planar trees.}%
\label{fig:cdt}
\end{figure}

Since every triangle in a causal triangulation is bordered by exactly one spacelike edge, there is a canonical pairing of triangles sharing a spacelike edge.
By joining all these pairs of triangles into faces, in other words, by deleting the spacelike edges, we end up with a quadrangulation of the sphere (figure \ref{fig:cdt1}). 
Notice that the removal of the spacelike edges has no effect on the distance labelling, because by construction they connect vertices with indentical labels.
Therefore the CDT time on the quadrangulation corresponds exactly to the labeling in the context of the Cori--Vauquelin--Schaeffer bijection.
The class of quadrangulations, which we call \emph{causal quadrangulations}, arising from this construction is easily seen to be characterized by the presence of a unique local maximum of the labeling.

If we root the pointed quadrangulation at one of the edges incident to the maximal vertex and apply Schaeffer's prescription, we end up with a rooted planar tree with all edges labeled by $-$'s (figure \ref{fig:cdttree}).
A direct consequence is that rooted causal quadrangulations with $N$ faces, and also rooted causal triangulations with $2N$ triangles, are counted by the Catalan numbers $C(N)$.
Their generating function is identical to the generating function $z^{(u)}(g)$ of unlabeled trees,
\begin{equation}\label{eq:unlabeledgenfun}
z^{(u)}(g) = \sum_{N=0}^{\infty} C(N)~g^N = \frac{1-\sqrt{1-4g}}{2g}.
\end{equation}
Later we will see how adding a coupling associated with the labeling allows us to get causal quadrangulations and unrestricted quadrangulations as special cases of a more general model of random trees.
But first let us show how one can extract more non-trivial scaling information from these representations by allowing quadrangulations with boundaries.

\subsection{Quadrangulations with a boundary}\label{sec:bijboundary}

As is shown in \cite{bouttier_distance_2009,bettinelli_scaling_2011,curien_uniform_2012} the Cori--Vauquelin--Schaeffer bijection extends in a natural way to quadrangulations with a boundary.
According to our definition in section \ref{sec:definitions} these are represented by planar maps of which all but one of the faces have degree 4.
Quadrangulations with one boundary necessarily have even boundary length, which we denote by $2l$.
An example of a quadrangulation with a boundary of length $28$ is shown in figure \ref{fig:boundary0}.
It is convenient to root the quadrangulation by selecting a corner of the boundary face (indicated by an arrow in figure \ref{fig:boundary0}).
As before, one of the vertices is marked as the origin, which may lie on the boundary. 

\begin{figure}[t!]
\centering
\subfloat[]{
\centering
\includegraphics[width=0.47\textwidth]{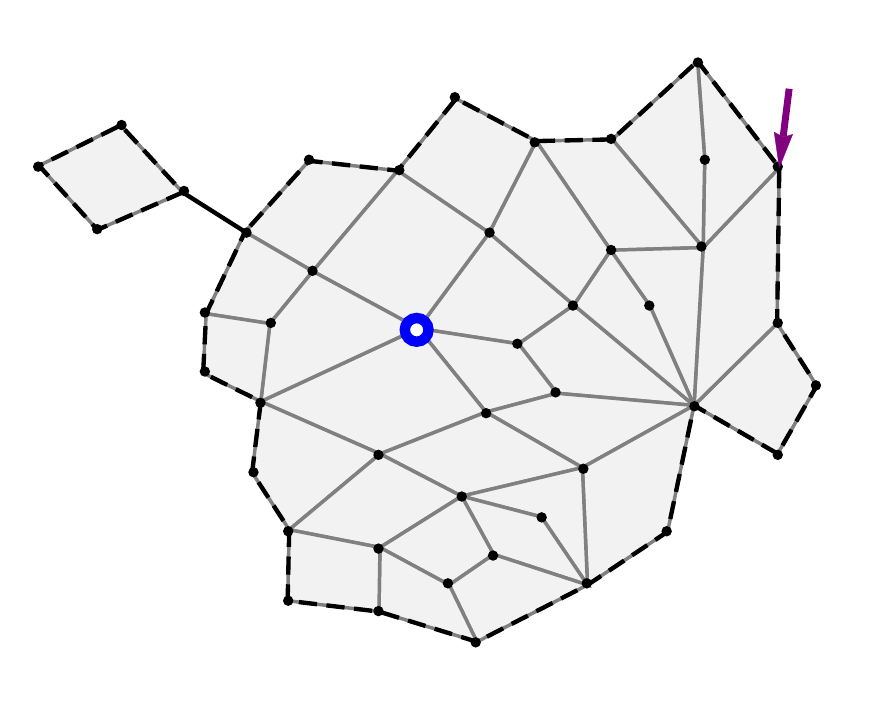}
\label{fig:boundary0}
}
\subfloat[]{
\centering
\includegraphics[width=0.47\textwidth]{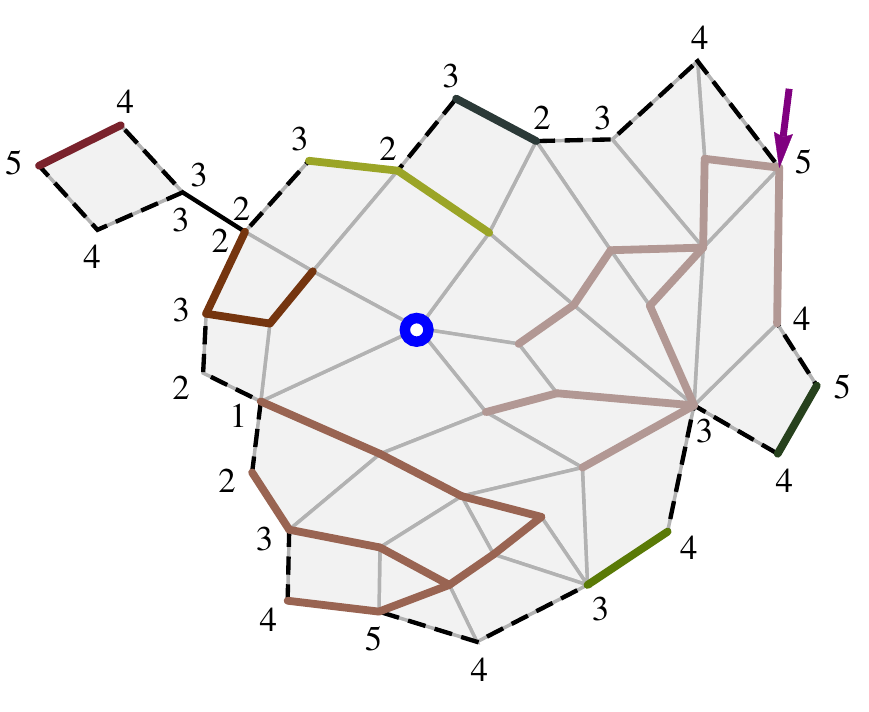}
\label{fig:boundary1}
}\\
\vspace{-0.6cm}
\subfloat[]{
\centering
\includegraphics[width=0.6\textwidth]{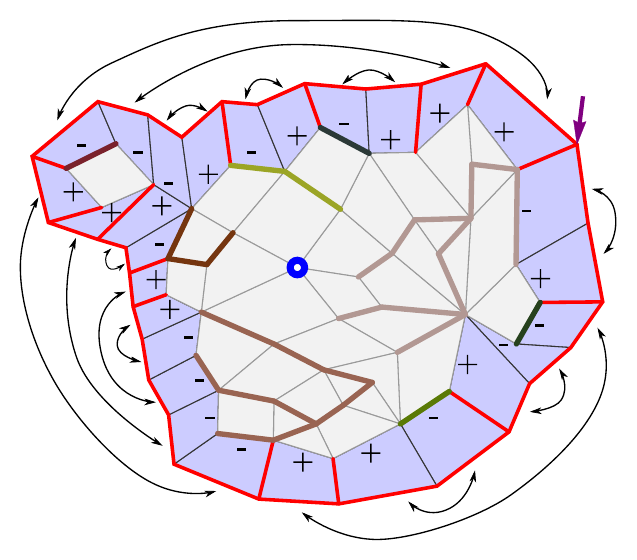}
\label{fig:boundary3}
}\\
\vspace{-0.6cm}
\subfloat[]{
\centering
\includegraphics[width=0.8\textwidth]{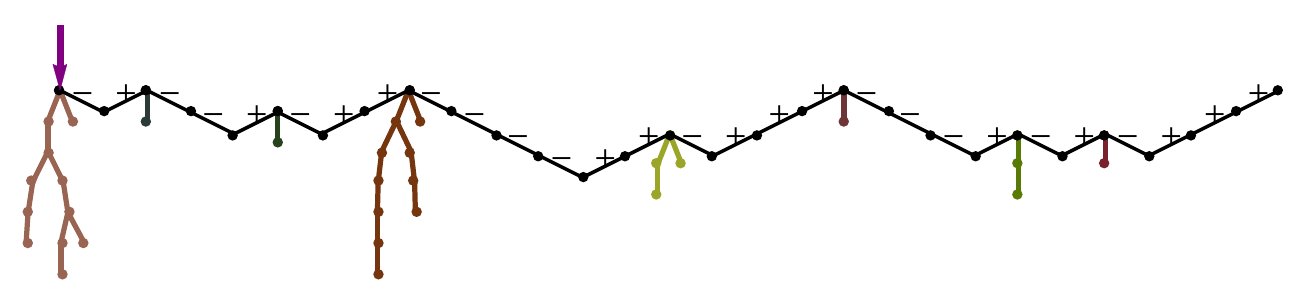}
\label{fig:boundary2}
}
\caption{Quadrangulations with a boundary and an origin are in bijection with sequences of $+$'s and $-$'s with a (possibly empty) labeled planar tree growing from each end-point of a $+$-edge.}%
\label{fig:boundary}
\end{figure}

Applying Schaeffer's prescription to the distance labeling we obtain a \emph{forest} $\mathcal{F}$, i.e. a set of disjoint trees (figure \ref{fig:boundary1}), instead of a single tree in the no-boundary case.
Let us orient the boundary of the quadrangulation in a clockwise direction.
Then to each boundary edge we can assign a $+$ or $-$ according to whether the label increases or decreases along the edge.
It turns out that each tree contains exactly one vertex that is the end-point of a $+$-edge.
This was shown in \cite{bouttier_distance_2009}, section 3.1, using a general bijection between planar maps and so-called \emph{labeled mobiles}.
To keep the discussion self-contained, let us present here an alternative argument based on the Cori--Vauquelin--Schaeffer bijection alone.

We turn the quadrangulation with boundary into a quadrangulation of the 2-sphere by quadrangulating the boundary face, while making sure that the distance labeling is unaffected.
A convenient method to do this is shown in figure \ref{fig:boundary3}.
First an annulus consisting of $2l$ quadrangles is glued to the boundary.
Notice that this does not change the distances from the original vertices to the origin and that the labels on the ``new boundary'' are increased by one compared to the ``old boundary''.
Now we only have to find a pairwise gluing of the new boundary edges in such a way that only vertices with identical label are identified.
A well-defined prescription is to repeatedly glue pairs consisting of a $-$-edge directly followed by a $+$-edge, until no boundary edges are left (see figure \ref{fig:boundary3} for an example).
Applying Schaeffer's prescription to the resulting pointed quadrangulation, we obtain a tree $\mathcal{T}$ which by construction contains the trees of $\mathcal{F}$ as subtrees.
The complement $\mathcal{T}\setminus\mathcal{F}$ of $\mathcal{F}$ in $\mathcal{T}$, as shown in red in figure \ref{fig:boundary3}, is a tree arising from the colouring of the quadrangles in the annulus. 
The quadrangles adjacent to the $-$-edges provide the colouring of the whole new boundary (after the gluing).
The quadrangles adjacent to the $+$-edges, on the other hand, lead to tree edges connecting the new boundary with all the ends of $+$-edges of the old boundary. 
Therefore only the ends of the $+$-edges are connected by the tree $\mathcal{T}\setminus\mathcal{F}$.  
Since $\mathcal{T}$ is connected and has no cycles, each tree of $\mathcal{F}$ must contain exactly one end of a $+$-edge.

One can root each tree of $\mathcal{F}$ at its distinguished vertex on the boundary by choosing the corner facing the external face.
In order to turn them into rooted well-labeled trees, we may shift the labels on the trees by an integer such that each of them has minimal label equal to one, or, equivalently, we only keep track of the $+$'s and $-$'s along the edges of the trees.
In general we may extract from a pointed quadrangulation rooted at its boundary of length $2l$ a sequence containing $l$ $+$'s and $l$ $-$'s, representing the change in the labeling along the boundary, and $l$ (possibly empty) rooted well-labeled trees which "grow" from the end of the $l$ $+$-edges (see figure \ref{fig:boundary2}).
It is not hard to see that one can reconstruct from this information the well-labeled tree $\mathcal{T}$ and therefore the full quadrangulation using the Cori--Vauquelin--Schaeffer bijection.
A precise proof is given in  \cite{bouttier_distance_2009,bettinelli_scaling_2011}.

Using this bijection one can easily write down a generating function $w(g,l)$ for the number of pointed quadrangulations with $N$ faces rooted on its boundary of length $2l$.
Since we can arrange the $l$ $+$'s in any way along the boundary of length $2l$, we get
\begin{equation}\label{eq:dtdiskl}
w(g,l) = \binom{2l}{l} z^{(\ell)} (g)^l,
\end{equation}
where $z^{(\ell)}(g)$ is the generating function for labeled trees 
(\ref{eq:labeledgenfun}).
In terms of the boundary cosmological constant $y$,
\begin{equation}\label{eq:dtdisk}
w(g,y)=\sum_{l=0}^{\infty} w(g,l) y^l = \frac{1}{\sqrt{1-4y z^{(\ell)}(g)}}.
\end{equation}

To obtain the continuum disk-function of 2d gravity we introduce a lattice spacing $\epsilon$, in terms of which we can define a continuum volume $V = N \epsilon^2$ and boundary length $L = l \epsilon$ with canonical dimension.
The Laplace transform $W_{\Lambda}(Y)$ of the continuum disk function $W_V(L)$ is obtained from the discrete disk function $w(g,y)$ by expanding around its critical point at $g_c=1/12$, $y_c=1/8$,
\begin{equation}
g=g_c(1-\Lambda \epsilon^2),\quad y=y_c(1-Y\epsilon).
\end{equation}
Plugging these into (\ref{eq:dtdisk}), we obtain
\begin{equation}
w(g,y) = \frac{1}{\sqrt{Y+\sqrt{\Lambda}}} \epsilon^{-1/2}(1+\mathcal{O}(\epsilon))=: W_{\Lambda}^{\text{m}}(Y) \epsilon^{-1/2}(1+\mathcal{O}(\epsilon)),
\end{equation}
where $W_{\Lambda}^{\text{m}}(Y)$ is the continuum disk function with a marked point.
The unmarked disk function (also known as the genus 0 sector of the Hartle--Hawking wave function of 2d gravity, see e.g. \cite{ambjorn_quantum_1997} section 4.4.2) is obtained by integrating $W_{\Lambda}^{\text{m}}$ with respect to $\Lambda$ (see also \cite{bouttier_distance_2009}, section 4.4),
\begin{equation}
W_{\Lambda}(Y) = \frac{2}{3}(Y-\frac{1}{2}\sqrt{\Lambda})\sqrt{Y+\sqrt{\Lambda}}.
\end{equation}

More generally one can consider any ensemble of labeled trees with generating function $z(g)$ and compute the correspondig disk function. 
Provided that the ensemble has a suscepitibility exponent $\gamma=1/2$, i.e. $z(g)$ is of the form
\begin{equation}
z(g_c) - z(g) \propto (g_c-g)^{1/2},
\end{equation}
we can define the continuum tree amplitude $Z$ through $z(g)=z(g_c)(1-Z\epsilon+\mathcal{O}(\epsilon^2))$.
The disk function is then simply 
\begin{equation}\label{eq:markeddisk}
W^{\mathrm{m}}(Y) = \frac{1}{\sqrt{Y+Z}}.
\end{equation}
In particular, one can consider the generating function $z^{(u)}(g)$ for the ensemble of trees where all labels are $-$'s, related to the causal quadrangulations.
The continuum amplitude is simply $Z^{(u)}=\sqrt{\Lambda}$, which is exactly the same as for the labeled trees.
Contrary to what one might have expected, restricting to causal quadrangulations does not change the continuum disk function. 

The explanation is that this particular disk function is not the one usually considered in the context of CDT.
To recover the latter one should consider the disk function where the boundary is restricted to be at constant distance from the origin.
In terms of the quadrangulations this is achieved by restricting the labels on the boundary to alternate between two consecutive integers, or, using the bijection, by fixing the boundary sequence to $(+,-,+,-,+,\ldots)$.
Since there is only one such sequence, we loose the combinatorial factor in (\ref{eq:dtdiskl}) and end up with the disk function
\begin{equation}
w_c(g,y)=\sum_{l=0}^{\infty} w(g,l) y^l = \frac{1}{1-y\, z(g)}.
\end{equation}
Its continuum counterpart is
\begin{equation}\label{eq:constdisk}
W^c_{\Lambda}(Y) = \frac{1}{Y+Z},
\end{equation}
which differs from the unrestricted disk function (\ref{eq:markeddisk}) only by an overall square root.
We reproduce the standard CDT disk function\footnote{To get the exact generating function for causal triangulations, one has to glue triangles to all $(-,+)$ pairs on the boundary. The generating function $w_{\mathrm{CDT}}(g,y)$ for the number of causal triangulations with a fixed number of triangles and a fixed boundary length is then $w_{\mathrm{CDT}}(g,y)=w_c(g^2,gy)$.} by setting $Z=\sqrt{\Lambda}$.

Notice that the overall square root in the marked DT disk function $W^{\mathrm{m}}_{\Lambda}(Y)$ compared to the CDT disk function $W^c_{\Lambda}(Y)$ is simply a consequence of the labeling describing a random walk on the boundary.
As we will see in section \ref{sec:timedep}, only once the random character of the labeling on the trees is taken into account, will the stark difference in scaling of DT compared to CDT be revealed.

In the next section we will introduce a new partition function for labeled trees, to which we can assign disk functions like above.

\section{From quadrangulations to generalized CDT}\label{sec:gencdt}

As mentioned earlier causal quadrangulations are characterized by a single local maximum of the labeling, while general quadrangulations can have any number.
This suggests a convenient way to interpolate between both models, by assigning a coupling $\g$ to each local maximum.
Then by construction setting $\g=0$ will lead to a model of closed causal quadrangulations and $\g=1$ to general quadrangulations.
If we interpret the labeling as a time function, we can also view the coupling $\g$ as a weight for the process of a universe splitting in two, i.e. of spatial topology change.

A model of spatial topology change in continuum CDT was studied in \cite{ambjorn_putting_2007,ambjorn_string_2008,ambjorn_matrix_2008} and was referred to as \emph{generalized CDT}.
To prevent a proliferation of baby universes in the continuum limit it was found that the coupling $\g$ should be scaled to zero with the lattice spacing $\epsilon$ as $\g=\g_s\epsilon^3$.
The continuum disk function could be calculated from a graphical consistency relation, leading to
\begin{equation}\label{eq:capfunction0}
W_{\lambda,\g_s}(X) = \frac{-(X^2-\lambda)+(X-\alpha)\sqrt{(X+\alpha)^2-2\g_s/\alpha}}{2\g_s},
\end{equation}
where $\lambda$ is the ``effective'' CDT cosmological constant (to be 
discussed in more detail below equation (\ref{eq:generalizedgscaling})) and 
$\alpha=\alpha(\g_s,\lambda)$ is given by the (largest) solution to
\begin{equation}\label{eq:alpha}
\alpha^3-\lambda\alpha+\g_s=0.
\end{equation}
This equation ensures that $X W_{\lambda,\g_s}(X) \to 1$ for $X \to \infty$. 
In the following we will show how we can derive $W_{\lambda,\g_s}(X)$  
by counting quadrangulations with a weight assigned to the local maxima of 
their labeling and then taking a suitable continuum limit.

However, before we continue, let us introduce a more convenient version of the disk amplitude in (\ref{eq:capfunction0}), for which one of the end points of the baby universes, i.e. one of the local maxima of the time function, is marked.
The term of order $\g_s^n$ in (\ref{eq:capfunction0}) corresponds to surfaces with $n+1$ such local maxima.
Hence we introduce the \emph{cap function}
\begin{equation}\label{eq:capfunction}
W^{\mathrm{cap}}_{\lambda,\g_s}(X) = \frac{1}{\sqrt{(X+\alpha)^2 - 2 \g_s/\alpha}},
\end{equation}
which satisfies
\begin{equation}
W^{\mathrm{cap}}_{\lambda,\g_s}(X) = 
\frac{\partial}{\partial \g_s}\left(\g_s  W_{\lambda,\g_s}(X) \right),
\end{equation}
where the derivative is taken while keeping $X$ and $\lambda$ fixed.

\begin{figure}[t]
\centering
\subfloat[]{
\centering
\includegraphics[height=5cm]{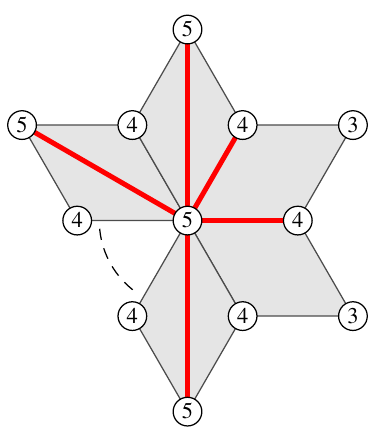}
\label{fig:maximum}
}
\subfloat[]{
\centering
\includegraphics[height=5cm]{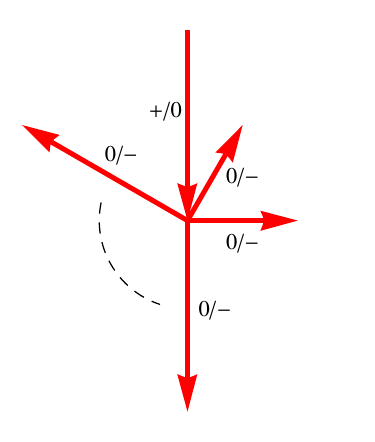}
\label{fig:maximum2}
}
\caption{A local maximum on the quadrangulation corresponds to a local maximum on the tree, i.e. to a vertex with a $+$- or $0$-edge coming in and all outgoing edges carrying a $0$ or a $-$.}%
\label{fig:maxfig}
\end{figure}

A vertex with label $t$ is a local maximum on the quadrangulation if all its neighbours, i.e. the vertices connected to it by an edge of the quadrangulation, are labeled $t-1$, as in figure \ref{fig:maximum}.
Equivalently, its neighbours in the associated well-labeled tree are labeled by $t$ or $t-1$.
Therefore, in terms of the rooted labeled trees we should associate a coupling $\g$ to every vertex in the tree which has a $+$- or $0$-edge coming in (or no edge in the case of the root vertex) and all outgoing edges carrying a $0$ or a $-$ (figure \ref{fig:maximum2}).
Let us denote by $z_0(g)=z_0(g,\g)$ the generating function for such trees.
Similarly we introduce the generating function $z_1(g)=z_1(g,\g)$ with the only difference that we do not assign a coupling $\g$ to the root vertex, even if there is local maximum there.
Both $z_0(g)$ and $z_1(g)$ therefore reduce to the generating function for rooted well-labeled trees $z^{(\ell)}(g)$ from (\ref{eq:labeledgenfun}) in the case $\g=1$.
We obtain recurrence relations for $z_0(g)$ and $z_1(g)$ by summing over the number and associated labels of the edges leaving the root,
\begin{align}
z_1 &= \sum_{k=0}^{\infty} \left(g\,z_1 +g\,z_0+g\,z_0\right)^k = \frac{1}{1-g\,z_1-2g\,z_0}\nonumber\\
z_0 &= \sum_{k=0}^{\infty} \left(g\,z_1 +g\,z_0+g\,z_0\right)^k + (\g-1)\sum_{k=0}^{\infty} \left(g\,z_1 +g\,z_0\right)^k \label{eq:zrecur}\\
&= z_1 + \frac{(\g-1)}{1-g\,z_1-g\,z_0}\nonumber
\end{align}
These combine into a single fourth-order polynomial for $z_1(g)$,
\begin{equation}\label{eq:zmin}
3 g^2\, z_1^4 - 4 g\, z_1^3 + (1+2g(1-2\g))z_1^2 - 1=0.
\end{equation}
The relevant solution, i.e. the one of the form $z_1(g)= 1+\mathcal{O}(g)$, is given by its smallest positive root.
For $\g=1$ the solution reduces to (\ref{eq:labeledgenfun}), while for $\g=0$ we get the generating function for unlabeled trees (\ref{eq:unlabeledgenfun}), $z_1(g)|_{\g=0} =  z^{(u)}(g)$.
The latter is not true for $z_0(g)$ since each labeled tree must have at least one local maximum, hence $z_0(g)=0$ for $\g=0$.

\begin{figure}[t]
\centering
\includegraphics[height=6.5cm]{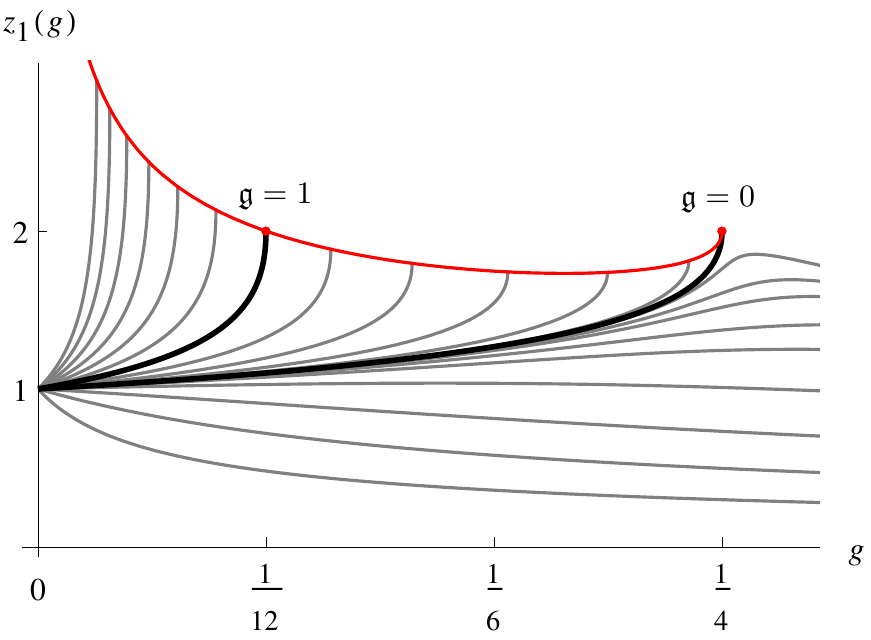}
\caption{The solution of eq.\ (\ref{eq:zmin}) for various values 
of $\g$ as a function of $g$. 
For each value of $\g\leq 0$ the curve ends at the critical line, defined 
by eq.\ (\ref{ja1}).}%
\label{fig:phaseplot}
\end{figure}

In figure \ref{fig:phaseplot} we have plotted the solutions for 
various other fixed values of $\g$.
For each $\g\geq 0$ there is a critical value $g_c(\g)$ at which $z_1(g)$ 
becomes non-analytic and at which a continuum limit can be taken. 
This critical value
$g_c(\g)$ is determined by the additional 
requirement that the derivative $z_1'(g)$
of $z_1(g)$ diverges, which 
leads to the equation
\begin{equation}\label{ja1}
3g_c^2 z_1^4-2 g_c z_1^3+1=0.
\end{equation}
This curve is plotted in red in figure \ref{fig:phaseplot}.
 
Unless $\g=0$ we obtain an infinite density of baby universes in the continuum limit and presumably we end up in the same universality class as pure DT.
This can be seen by calculating the expected number $\langle N_{\mathrm{max}}(\g)\rangle_N$ of local maxima for fixed large number $N$ of faces, which satisfies
\begin{equation}
\langle N_{\mathrm{max}}(\g)\rangle_N = \rho(\g) N + \mathcal{O}(N^0), \quad \rho(\g) = 2 \left(\frac{\g}{2}\right)^{2/3} + \mathcal{O}(\g), \quad \rho(1) = \frac{1}{2}.
\end{equation}

If we want to keep $\langle N_{\mathrm{max}}(\g)\rangle_N$ finite as $N\to\infty$ we should instead scale $\g$ to zero like $N^{-3/2}$.
In the grand-canonical setting this corresponds to scaling $\g = \g_s \epsilon^3$ with the lattice spacing $\epsilon$, as observed previously.
To take this continuum limit we again expand $g$ 
around its critical value $g_c(\g) = 1/4 - 3/4 (\g/2)^{2/3} + \mathcal{O}(\g)$,
\begin{equation}\label{eq:generalizedgscaling}
g = g_c(\g)(1-\Lambda \epsilon^2) = \frac{1}{4} \left( 1 - 3 \left(\frac{\g_s}{2}\right)^{2/3}\epsilon^2 - \Lambda\epsilon^2\right) = \frac{1}{4} (1 - \lambda \epsilon^2),
\end{equation}
where $\lambda = \Lambda + 3 \left(\frac{\g_s}{2}\right)^{2/3}$ is the 
``effective''  cosmological constant as it was introduced 
originally in generalized CDT \cite{ambjorn_putting_2007}. The original model
was formulated directly in the continuum limit and 
only later, using matrix models, was it understood that $\lambda$
is actually a sum of contributions coming from a ``genuine'' cosmological 
constant $\Lambda$ related to the ``area'' of graphs (triangulations, 
quadrangulations etc.) and a  (string) coupling constant $\g_s$ 
from splitting off baby 
universes \cite{ambjorn_new_2012}.  

Plugging (\ref{eq:generalizedgscaling}) and $z_1(g) = z_1(g_c)(1-Z_1\epsilon)$ into (\ref{eq:zmin}) we obtain the continuum equation
\begin{equation}\label{jax}
Z_1^3 - \lambda Z_1 + \g_s = 0.
\end{equation}
Notice that this is exactly equation (\ref{eq:alpha}) with $\alpha = Z_1$.

According to (\ref{eq:constdisk}) the disk function with constant distance to the origin associated to this ensemble of labeled trees is given by
\begin{equation}\label{eq:gencdtcup}
W_{\lambda,\g_s}^{\mathrm{cup}}(Y) = \frac{1}{Y+Z_1},
\end{equation}
which is different from (\ref{eq:capfunction}).
The difference is in the distance function used for the labeling. 
Here the labeling corresponds to the distance to the marked origin in the disk, while in (\ref{eq:capfunction}) the distance to the boundary is used and one of the local maxima of the time function is marked.
We refer to these two disk functions as the \emph{cup function} 
$W_{\lambda,\g_s}^{\mathrm{cup}}(Y)$ and 
the \emph{cap function} $W_{\lambda,\g_s}^{\mathrm{cap}}(X)$, 
see figure \ref{fig:capcup}. In order to obtain 
the cap function we will now study more general time-dependent amplitudes,
in particular the two-loop propagator   $G_{\lambda,\g_s}(X,Y;T)$ (figure \ref{fig:capcup}c) 
and the two-point function  $G_{\lambda,\g_s}(T)$. 

\begin{figure}[t]
\centering
\includegraphics[width=0.95\textwidth]{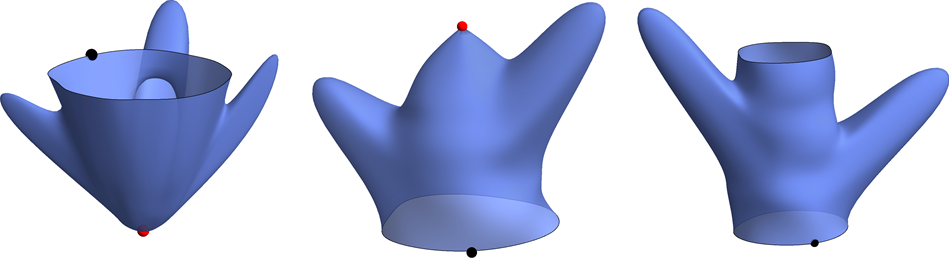}

(a)\hspace{4.4cm}(b)\hspace{4.4cm}(c)
\caption{(a) The cup function $W^{\mathrm{cup}}_{\lambda,\g_s}(Y)$, 
(b) the cap function $W^{\mathrm{cap}}_{\lambda,\g_s}(X)$ and 
(c) the propagator $G_{\lambda,\g_s}(X,Y;T)$. }%
\label{fig:capcup}
\end{figure}

\section{Time-dependent amplitudes}\label{sec:timedep}

The simple generating functions discussed in the previous section provide little information about the geometry of the quadrangulations they encode. 
In order to better understand the geometries and to fully reproduce the results of \cite{ambjorn_putting_2007} we need to keep track of the labeling on the trees in more detail.
Let us define the generating function $z_0(t)=z_0(t,g,\g)$ for rooted trees with positive integer labels on its vertices, such that the labels differ by at most one along the edges and such that the root is labeled $t$.
Then $z_0(t)-z_0(t-1)$ is the generating function for rooted well-labeled trees, i.e. labeled trees with minimal label equal to 1, with root labeled $t$.
According to the bijections in section \ref{sec:bijections}, $z_0(t)$ also gives a generating function for the quadrangulations of the sphere with the furthest end-point of its root edge at most a distance $t$ from the origin, again including a factor of $\g$ for each local maximum of the distance functions.
Likewise, the generating function where this distance is exactly $t$ is $z_0(t)-z_0(t-1)$. 

As in the previous section we introduce $z_1(t)=z_1(t,g,\g)$ for which no coupling $\g$ is assigned to the root vertex.
The recurrence relations (\ref{eq:zrecur}) straightforwardly generalize to
\begin{equation}\label{eq:ztrecur}
\begin{aligned}
z_1(t) & = \frac{1}{1-g\, z_1(t-1)-g\, z_0(t)-g\, z_0(t+1)}, \\
z_0(t) & = z_1(t) + \frac{(\g-1)}{1-g\, z_1(t-1)-g\, z_0(t)}
\end{aligned}
\end{equation}
for $t\geq 1$, subject to the boundary conditions $z_1(0)=0$ and $z_0(\infty)=z_0$.

Quite remarkably, (\ref{eq:ztrecur}) can be solved analytically using the technique outlined in \cite{bouttier_geodesic_2003,di_francesco_geodesic_2005}.
Expanding the generating functions around their limits as $t\to\infty$, which are given by the solution $z_1$ and $z_0$ to (\ref{eq:zrecur}), we find after a straightforward but tedious calculation
\begin{equation}\label{eq:ztsolution}
\begin{aligned}
z_1(t) & = z_1 \, \frac{1-\sigma^t}{1-\sigma^{t+1}} \, \frac{1-(1-\beta)\sigma-\beta \sigma^{t+3}}{1-(1-\beta)\sigma-\beta \sigma^{t+2}}, \\
z_0(t) &= z_0 \, \frac{1-\sigma^t}{1-(1-\beta)\sigma-\beta \sigma^{t+1}} \, \frac{(1-(1-\beta)\sigma)^2-\beta^2\sigma^{t+3}}{1-(1-\beta)\sigma-\beta \sigma^{t+2}},
\end{aligned}
\end{equation}
where $\beta$ and $\sigma$ are fixed in terms of $z_0$ and $z_1$ (hence in terms of $g$ and $\g$) through
\begin{equation}\label{eq:betasigma}
\begin{aligned}
g(1+\sigma)(1+\beta\sigma)z_1 - \sigma (1-2g\,z_0) &=0,\\
(1-\beta)\sigma - g(1+\sigma) z_1 + g(1-\sigma+2\beta\sigma)z_0 &=0.
\end{aligned}
\end{equation}
In particular, $\beta=0$ in the case of CDT ($\g=0$) and $\beta=1$ in the case of DT ($\g=1$).

To get to generalized CDT in the continuum limit, the time should be scaled canonically with the lattice spacing, i.e. $t = T/ \epsilon$. 
The scaling $\g = \g_s \epsilon^3$ and $g = 1/4 (1-\lambda \epsilon^2)$ implies that the parameters $\sigma$ and $\beta$ scale as $\sigma = 1 - 2\Sigma \epsilon$ and $\beta = B \epsilon$. 
According to (\ref{eq:betasigma}) the continuum parameters $\Sigma$ and $B$ are related to $\lambda$ and $\g_s$ through
\begin{equation}\label{eq:betasigmacont}
\begin{aligned}
\g_s - 2B(B^2 + 3 B \Sigma + 2\Sigma^2) &=0, \\
\lambda - 3 B^2 - 6B \Sigma - \Sigma^2 &=0.
\end{aligned}
\end{equation}
They can also be expressed in terms of $\g_s$ and $\alpha$ from equation (\ref{eq:alpha}),
\begin{equation}
\begin{aligned}
\Sigma &= \frac{1}{2}\sqrt{ 4\alpha^2 - 2 \g_s/\alpha } \quad \left(= \sqrt{\lambda} - \frac{3\g_s}{4\lambda} + \mathcal{O}(\g_s^2)\right), \\
B &= \alpha - \Sigma.
\end{aligned}
\end{equation}
We also note that, as $\g\to\infty$, $\alpha$ and $\Sigma$ grow as
\begin{equation}\label{ja7}
\alpha = \left(\frac{\g_s}{2}\right)^{1/3} +  \left(\frac{\Lambda}{3}\right)^{1/2}
+ O\left(\g_s^{-1/3}\right),~~~~
\Sigma = (3\Lambda)^{1/4} \left(\frac{\g_s}{2}\right)^{1/6} + 
O\left(\g_s^{-1/6}\right) 
\end{equation}
Finally, the solutions (\ref{eq:ztsolution}) scale as $z_1(t) = 2(1-Z_1(T)\epsilon)$ and $z_0(t)=2 Z_0(T) \epsilon^2$ with
\begin{equation}\label{eq:contztsol}
\begin{aligned}
Z_1(T) &= \alpha+\frac{\Sigma^2}{\sinh(\Sigma T)\left[\Sigma\cosh(\Sigma T)+\alpha\sinh(\Sigma T)\right]}, \\
Z_0(T) &= \frac{\g_s}{2\alpha}\left(1-\frac{\Sigma^2}{\left[\Sigma\cosh(\Sigma T)+\alpha\sinh(\Sigma T)\right]^2}\right).
\end{aligned}
\end{equation}
Since $z_0(t)-z_0(t-1)$ defines a discrete two-point function, the scaling limit of the two-point function is obtained by differentiating $Z_0(T)$ with respect to $T$,
\begin{equation}\label{ja10}
G_{\lambda,\g_s}(T) = \frac{\rmd Z_0(T)}{\rmd T} = \Sigma^3 \,\frac{\g_s}{\alpha}\,
\frac{\Sigma \sinh \Sigma T +\alpha \cosh \Sigma T}{ 
\Big(\Sigma \cosh \Sigma T +\alpha \sinh \Sigma T\Big)^3}.
\end{equation}
In the limit $\g_s \to 0$ we obtain (up to a factor $\g_s$ which is 
convention) the CDT result 
\begin{equation}\label{ja11}
G_{\lambda,\g_s=0}(T) \sim \rme^{-2\sqrt{\Lambda} T}. 
\end{equation}
In the limit $\g_s \to \infty$ we obtain (again up to a $\g_s$ factor)
\begin{equation}\label{ja12}
G_{\lambda,\g_s\to\infty}(T) \sim \Lambda^{3/4}
\frac{\cosh ( \Lambda^{1/4}  T')}{\Big(\sinh(
\Lambda^{1/4}  T')\Big)^3 }, ~~~~T' = 3^{1/4} \left(\frac{\g_s}{2}\right)^{1/6}T.
\end{equation}
Of course this limit does not exist unless we keep 
$T'$ finite in the scaling limit rather than $T$. 
Thus we are really discussing another scaling limit!
Recall that $T \sim t/\epsilon$ and $\g_s= \g/\epsilon^3$. Thus a 
finite $T'$ in the limit $\epsilon \to 0$ can be identified with a scaling 
limit where $t/\epsilon^{1/2}$ is finite for $\epsilon \to 0$. This is 
precisely the limit first discussed by Kawai et al. \cite{kawai_transfer_1993} where
the \emph{geodesic} distance scales anomalously with respect to volume,
leading to the Hausdorff dimension $d_h=4$ for 2d Euclidean 
quantum gravity and in the DT ensemble of random graphs. The 
two-point function (\ref{ja12}) was first calculated in \cite{ambjorn_scaling_1995}
(see also \cite{ambjorn_fractal_1995} and \cite{aoki_operator_1996} for more general formulas).
We note that the ``DT limit'' $\g_s \to\infty$ can be identified with 
the limit $\lambda /(\g_s/2)^{2/3} \to 3$ from above and is the limit 
where eq.\ (\ref{jax}) ceases to have a real positive solution $Z_1$.

\subsection{The disk functions and propagator}\label{subsec:boundaries}

Let us now turn to surfaces with one or more boundary components. 
We will first present a heuristic derivation of the disk function $W_{\lambda,\g_s}(X)$ in (\ref{eq:capfunction0}) using only the continuum tree amplitudes $Z_0(T)$ and $Z_1(T)$.
Below we will see that $W_{\lambda,\g_s}(X)$ can also be obtained directly in the scaling limit of quadrangulations with a boundary, by relating it to the propagator with a final boundary of length 0.

The cup function $(Y+Z_1(T))^{-1}$ constructed from $Z_1(T)$ 
corresponds to a path integral over all surfaces with a boundary at constant 
distance smaller than $T$ from the origin.
The corresponding cup function for fixed boundary length $L$ is its inverse Laplace transform and equals
\begin{equation}\label{eq:cupfunctionofLT}
\rme^{-L\,Z_1(T)}.
\end{equation}
We note that this function can be used to reveal 
a relation between $Z_0(T)$, $Z_1(T)$ and the disk function
$W_{\lambda,\g_s}(X)$ (which is the cap function without a mark in (\ref{eq:capfunction0}) and which we will determine by combinatorial methods below). 
\begin{figure}[t]
\centering
\includegraphics[height=5.5cm]{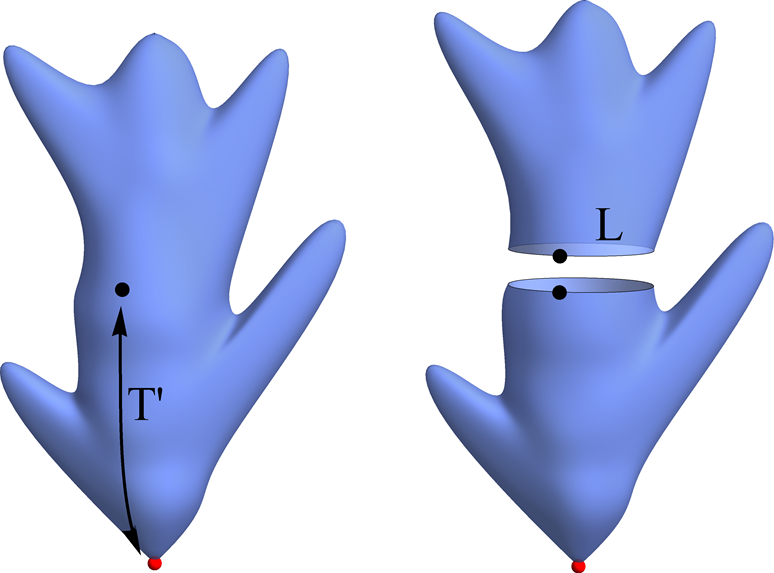}
\caption{The two-point function $Z_0(T)$ can be obtained by ``closing off'' the boundary of a cup function at a fixed distance $T'<T$ from the origin by the disk amplitude.}%
\label{figjan1}
\end{figure}
In figure \ref{figjan1} we have depicted the two-point function $Z_0(T)$, which is a path integral over surfaces with two marked points, i.e. the origin and the root, separated by a distance smaller than $T$.
This two-point function can be decomposed as follows: first one moves from the origin to the root at a distance $T'\leq T$ from the origin. This root is located on a connected curve of some length $L$ where all 
points have the same distance $T'$ to the origin. 
Cutting the surface along this curve leads to two disks, both with a mark on the boundary.
The bottom disk (see figure \ref{figjan1}) has an origin at a fixed distance $T'\leq T$ from the boundary and therefore its amplitude is given by (\ref{eq:cupfunctionofLT}).
The top disk has no mark in its interior and the labeling is determined by the distance to the boundary of length $L$, hence its amplitude corresponds to the inverse Laplace transform $\tilde{W}_{\lambda,\g_s}(L)$ of the disk function $W_{\lambda,\g_s}(X)$.
The two-point function $Z_0(T)$ is therefore obtained by combining these to disk amplitudes and integrating over $L$, i.e.
\begin{equation}\label{eq:Z0TfromZ1T}
Z_0(T) = \g_s\int_0^{\infty} \rmd L \, \tilde{W}_{\lambda,\g_s}(L) e^{-L\, Z_1(T)} = \g_s W_{\lambda,\g_s}(Z_1(T)).
\end{equation}
We have to include a factor of $\g_s$ simply because $Z_0(T)$ has a coupling $\g_s$ for each local maximum, while the term of order $\g_s^n$ in the disk function $W_{\lambda,\g_s}(X)$ corresponds to surfaces with $n+1$ local maxima (see the discussion above eq. (\ref{eq:capfunction})).
Since we have explicitly determined $Z_0(T)$ and $Z_1(T)$ we can verify using (\ref{eq:Z0TfromZ1T}) that $W_{\lambda,\g_s}(X)$ is indeed given by formula 
(\ref{eq:capfunction0}).

The \emph{propagator} $G_{\lambda,\g_s}(X,Y;T)$ corresponds to the path integral over surfaces with two boundaries, an initial boundary with boundary cosmological constant $X$ and a final boundary at a fixed distance $T$ from the initial boundary with boundary cosmological constant $Y$ (see figure \ref{fig:capcup}c).
In order to find an expression for $G_{\lambda,\g_s}(X,Y;T)$ in the scaling limit of generalized CDT a suitable ensemble of quadrangulations has to be chosen.
Let us define $\mathcal{Q}_{l_1,l_2,t}(N)$, $t\geq 1$, to be the set of quadrangulations $Q$ with $N$ quadrangles and two boundaries, the ``initial'' and ``final'' boundary, of lengths\footnote{Even though quadrangulations with two boundaries of \emph{odd} length exist, we require the boundary lengths to be even to keep them bipartite.} $2l_1$ and $2l_2$, rooted at the final boundary and satisfying the following conditions.
First of all we require that the initial boundary is \emph{non-degenerate} in the sense that all the corners of the boundary face belong to distinct vertices.
In order to state the second requirement a canonical labeling of the vertices is introduced.
If $t$ is odd, the vertices on the initial boundary that have even distance to the root are labeled $0$.
Otherwise the vertices with odd distance are labeled $0$.
All remaining vertices are labeled by their minimal distance to the set of vertices labeled $0$.
Given this canonical labeling, the vertices on the final boundary are required to have label $t$ or $t+1$, which is a discrete implementation of the fixed distance between the boundaries which we are after in the continuum. 
An example of a quadrangulation in $\mathcal{Q}_{l_1,l_2,t}(N)$ with $N=56$, $l_1=5$, $l_2=13$ and $t=3$ is shown in figure \ref{fig:propagator1}.

We can close off the initial boundary by gluing to it a disk 
constructed from $l_1$ simple faces, i.e. each of them labeled $(-1,0,1,0)$ (see figure \ref{fig:propagator2}).
This way we obtain a pointed quadrangulation, for which the origin is labeled $-1$ instead of the usual $0$, with a single boundary.
According to the bijection discussed in section \ref{sec:bijboundary} (and taking into account an overall shift of the labels by $-1$), this quadrangulation can be encoded by a forest of $l_2$ rooted trees, one ``growing'' from each vertex labeled $t+1$ on the final boundary (see figure \ref{fig:propagator2}).
The vertices of the trees carry non-negative labels, differing by at most one along the edges and having label $t+1$ on the root.
The quadrangulations under consideration have a special structure at the origin, namely that all $l_1$ faces around the origin are simple and that the $l_1$ edges starting at the origin end at distinct vertices.
In terms of the rooted trees with their edges directed away from the root, this condition translates exactly into the condition that the vertices labeled $0$ have no outgoing edges. 
The number of vertices labeled $0$ in the whole forest is $l_1$.

\begin{figure}[t]
\centering
\subfloat[]{
\centering
\includegraphics[width=7.5cm]{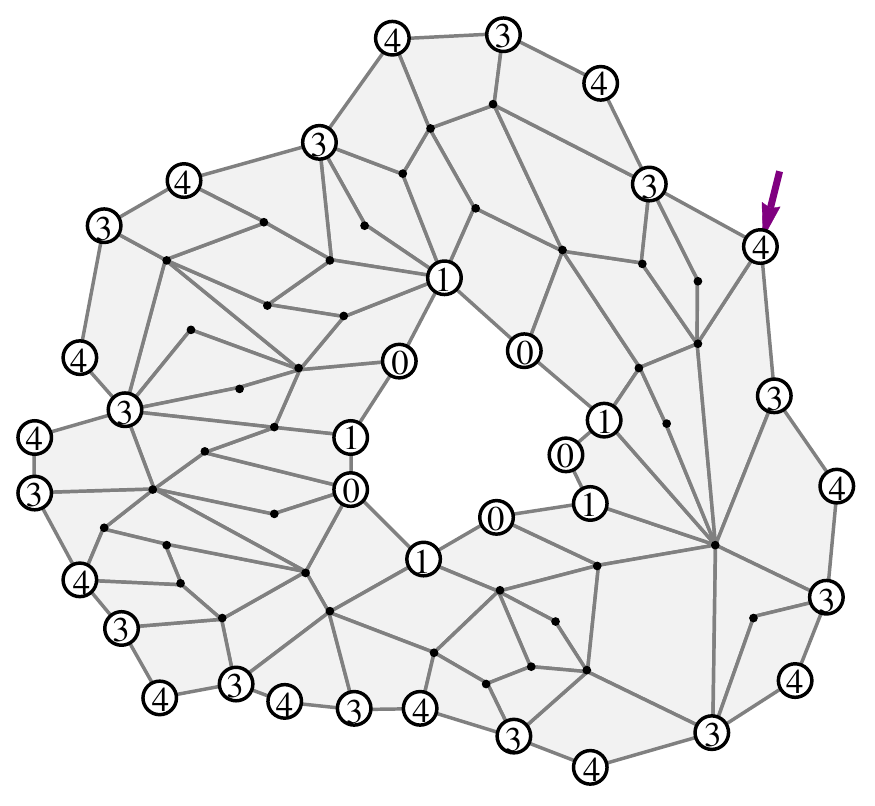}
\label{fig:propagator1}
}
\subfloat[]{
\centering
\includegraphics[width=7.5cm]{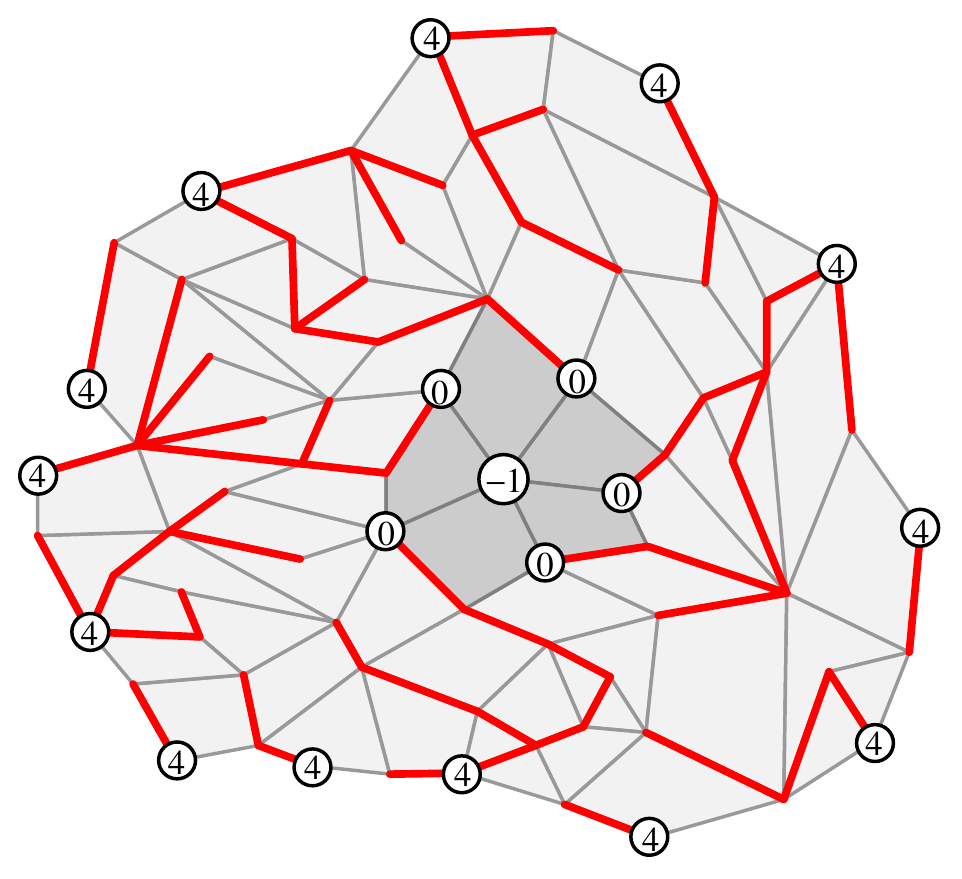}
\label{fig:propagator2}
}
\caption{(a) A quadrangulation with two boundaries contributing to the discrete propagator. (b) The initial boundary face can be quadrangulated by inserting a new vertex with label $-1$, turning the quadrangulation into a pointed quadrangulation with a single boundary. }%
\label{fig:propagator}
\end{figure}

In order to count these quadrangulations we introduce the generating functions $z_{i,x}(t)=z_{i,x}(t,g,\g)$, $i=0,1$, for rooted trees labeled by non-negative integers which differ by at most one along the edges, having label $t$ on the root, no edges starting at a vertex labeled $0$, carrying a factor of $x/g$ for each vertex labeled $0$ and a factor of $\g$ for each local maximum (except for the one at the origin in the case of $z_{1,x}(t)$).
When $x=0$ the labels are restricted to be positive and therefore $z_{i,x=0}(t)=z_i(t)$, where $z_i(t)$ are the generating functions from before.
It is not hard to see that for $x>0$ the $z_{i,x}(t)$ satisfy exactly the same equations (\ref{eq:ztrecur}) but with the boundary condition $z_{1,x}(0)=x/g$ instead of $z_1(0)=0$.

The generating function $\mathcal{G}_{\g}(x,y,g;t)$ for the number of quadrangulations in $\mathcal{Q}_{l_1,l_2,t}(N)$, including a factor $\g$ for each local maximum away from the final boundary, is then given by
\begin{equation}\label{eq:propagatorgenfun}
\begin{aligned}
\mathcal{G}_{\g}(x,y,g;t) & = \sum_{l_1=1}^{\infty}\sum_{l_2=1}^{\infty}\sum_{N=0}^{\infty} x^{l_1}y^{l_2} g^N\sum_{Q\in\mathcal{Q}_{l_1,l_2,t}(N)}  \g^{N_{\text{max}}(Q)} \\
&= \sum_{l_2=1}^{\infty}y^{l_2}\left[(z_{1,x}(t+1))^{l_2}-(z_1(t+1))^{l_2}\right],
\end{aligned}
\end{equation}
where the $x$-independent part $(z_1(t+1))^{l_2}$ has been subtracted on the right-hand to ensure $l_1 \geq 1$.

Since the solution (\ref{eq:ztsolution}) solves (\ref{eq:ztrecur}) for any \emph{real} value $t>0$ and $z_1(t)$ increases monotonically from $0$ at $t=0$ to $z_1$ at $t\to\infty$, we find (with a slight abuse of notation)
\begin{equation}
z_{1,x}(t) = z_1\left( t + z_1^{-1}(x/g)\right),
\end{equation}
provided $0\leq x/g < z_1$.
This solution has a critical point at $x=g z_1$, around which $z_{1,x}(t)$ should be expanded to get the canonical scaling of the initial boundary length in the continuum limit. 
Setting $x = 1/2(1 - X \epsilon + \mathcal{O}(\epsilon^2))$ leads to $z_{1,x}(t) = 2(1-Z_{1,X}(T)\epsilon+ \mathcal{O}(\epsilon^2))$ with
\begin{equation}
Z_{1,X}(T)=Z_1( T + Z_1^{-1}(X) ).
\end{equation}
However, one has to keep in mind that this solution is only valid for $X > Z_1 = \alpha$.
To get the other part of the solution we have to solve (\ref{eq:ztrecur}) with $\bar{z}_1(t)>z_1$, which can be formally obtained by shifting $t \to t + i \pi/(\log\sigma)$.
The corresponding continuum solution $\bar{Z}_1(T)$ is obtained from (\ref{eq:contztsol}) by shifting $T \to T - i \pi/(2\Sigma)$.
This $\bar{Z}_1(T)$ grows monotonically from $\sqrt{2\g_s/\alpha}-\alpha$ to $\alpha$ and therefore allows us to construct $Z_{1,X}(T)$ for $\sqrt{2\g_s/\alpha}-\alpha < X < \alpha$.
Notice that $X > \sqrt{2\g_s/\alpha}-\alpha$ is exactly the region where (\ref{eq:capfunction}) is real and non-singular.

If in addition the final boundary cosmological constant is scaled as $y=1/2(1-Y\epsilon+\mathcal{O}(\epsilon^2))$, it follows from (\ref{eq:propagatorgenfun}) that
\begin{equation}\label{eq:almostprop}
\mathcal{G}_{\g}(x,y,g;t) = \left(\frac{1}{Y+Z_{1,X}(T)}-\frac{1}{Y+Z_1(T)}\right)\epsilon^{-1} + \mathcal{O}(\epsilon^0).
\end{equation}
The term in the parentheses is almost the sought-after propagator $G_{\lambda,\g_s}(X,Y;T)$, except that the final boundary is marked, while we want only the initial boundary to be marked 
(see figure \ref{fig:capcup}c).
To get $G_{\lambda,\g_s}(X,Y;T)$ a factor of $l_1/l_2$ should be included in the sum (\ref{eq:propagatorgenfun}), or, equivalently, one can differentiate (\ref{eq:almostprop}) with respect to $X$ and integrate with respect to $Y$,
\begin{equation}
G_{\lambda,\g_s}(X,Y;T) = \int_{\infty}^Y \rmd Y' \frac{\partial}{\partial X} \left(\frac{1}{Y' + Z_{1,X}( T )}\right) 
= \frac{Z_{1,X}'(T)}{Z_{1,X}'(0)} \frac{1}{Y + Z_{1,X}( T )}.
\label{eq:propagatorsol}
\end{equation}
The cap function is obtained from the propagator by integrating over $T$ and taking the final boundary length to zero,
\begin{equation}\label{eq:capsolution}
W_{\lambda,\g_s}^{\mathrm{cap}}(X) = \lim_{Y\to\infty} Y \int_0^{\infty} \rmd T\,G_{\lambda,\g_s}(X,Y;T) = \frac{\alpha - X}{Z_1'(Z_1^{-1}(X))}.
\end{equation}
From (\ref{eq:contztsol}) it follows that
\begin{equation}\label{ja3}
Z_{1,X}'(T) = 
(\alpha - Z_{1,X}(T))\sqrt{(Z_{1,X}(T)+\alpha)^2 - 2\g_s/\alpha} =: 
-\bar{W}(Z_{1,X}(T)),
\end{equation}
where we have introduced the notation
\begin{equation}
\bar{W}(X) = (X-\alpha)\sqrt{(X+\alpha)^2-2\g_s/\alpha},
\end{equation}
and thus (\ref{eq:capsolution}) reproduces exactly (\ref{eq:capfunction}).
Further, (\ref{eq:propagatorsol}) shows that the 
loop propagator satisfies the differential equation
\begin{equation}\label{ja4}
\frac{\partial}{\partial T} \,G_{\lambda,\g_s}(X,Y;T) = 
-\frac{\partial}{\partial X}\Big( \bar{W}(X) G_{\lambda,\g_s}(X,Y;T)\Big),
\end{equation}
since (\ref{ja3}) is the characteristic equation for (\ref{ja4})
and (\ref{eq:propagatorsol}) is thus the solution with initial 
boundary condition 
\begin{equation}
G_{\lambda,\g_s}(X,Y;T=0)= \frac{1}{X+Y},~~~\mathrm{or}~~~
\tilde{G}_{\lambda,\g_s}(L_1,L_2;T=0) = \delta (L_1-L_2),
\end{equation}
where $\tilde{G}_{\lambda,\g_s}(L_1,L_2;T)$ is the inverse Laplace 
transform of $G_{\lambda,\g_s}(X,Y;T)$.
 
An alternative route towards the solutions (\ref{eq:contztsol}) 
and the propagator (\ref{eq:propagatorsol}) is to directly take the continuum limit of the recurrence relations (\ref{eq:ztrecur}).
This is done by substituting
\begin{equation}
\begin{aligned}
z_1(t+s) &= 2 \left( 1 + \epsilon (Z_1(T) + s \,\epsilon Z_1'(T)) \right) \\
z_0(t+s) &= 2 \epsilon^2 \left( Z_0(T) + s\, \epsilon Z_0'(T) \right)
\end{aligned} \quad s \in \{0,1,2,3\},
\end{equation}
in (\ref{eq:ztrecur}) leading to the differential equations
\begin{align}\label{ja30}
Z_1'(T) &= \lambda - Z_1(T)^2 - 2 Z_0(T), \\
Z_0'(T) &= 2 Z_1(T) Z_0(T) - \g_s.
\end{align}
It can be checked that these equations are solved by (\ref{eq:contztsol}),
but it might be more enlightening to show how they contain all
information about the disk function. First we note that
the equations imply 
\begin{equation}
Z_1''(T) = 2\left(Z_1(T)^3 -\lambda Z_1(T) +\g_s\right).
\end{equation}  
Integration and the assumption that $Z_1(T) \to Z_1=\alpha$ for $T \to 
\infty$ then implies:
\begin{equation}
Z_1'(T) =-\bar{W}(Z_1(T)),~~~~Z_0(T) = \g_s W_{\lambda,\g_s}(Z_1(T)). 
\end{equation}
where  $W_{\lambda,\g_s}(X)$ is given by (\ref{eq:capfunction0}). While 
this determines  $W_{\lambda,\g_s}(X)$ algebraically, one has to 
appeal to figure \ref{figjan1}, say, to identify it as the disk function, as done previously.

\section{Generalized CDT in terms of general planar maps}\label{sec:planarmaps}

Above we have seen that pointed quadrangulations with a certain number of local maxima of the distance functions can be encoded in well-labeled planar trees with the same number of local maxima. 
As we will demonstrate in this section, a set of planar maps including faces of any degree can be used to encode the same information.

There exists a well-known bijection $\Phi_0$, often referred to as the \emph{trivial bijection} (see e.g. \cite{gall_scaling_2011}), between pointed quadrangulations with $N$ faces and pointed planar maps with $N$ edges. 
It can be formulated in a very similar way as the Cori--Vauquelin--Schaeffer bijection, namely in terms of the distance labeling from the origin.
In stead of applying the Schaeffer's prescription in figure \ref{fig:quad0}a, for each face the diagonal is drawn that connects vertices of even label.
In this way all vertices with even label will be part of the planar map, and the vertices with odd label are in one-to-one correspondence with the faces of the map.
It can be seen that this planar map with marked origin completely characterizes the quadrangulation.

However, as we will see below, there exists another inequivalent bijection between rooted pointed quadrangulations $Q$ with $N$ faces and rooted pointed planar maps with $N$ edges. 
Although we expect this bijection to be known, since it is quite similar to the bijection introduce by Miermont in \cite{miermont_tessellations_2009} and used in \cite{bouttier_three-point_2008}, we were unable to find an explicit reference to it in the literature.
The bijection is a consequence of the general result in theorem \ref{thm:genbijection} below, which is quite similar to theorem 4 in \cite{miermont_tessellations_2009}.
To keep the discussion self-contained we provide a proof, which is largely inspired by the proofs in \cite{chassaing_random_2004,miermont_tessellations_2009}.

Let $\mathcal{Q}^{(l)}$ be the set of rooted quadrangulations of the sphere, i.e. quadrangulations with a marked unoriented edge, equipped with integer labels on the vertices, such that the labels differ by exactly one along the edges.
Let $\mathcal{M}^{(l)}$ be the set of rooted planar maps with integer labels on the vertices, such that the labels differ by \emph{at most} one along the edges.
Given a labeled rooted quadrangulation $Q \in \mathcal{Q}^{(l)}$, according to lemma \ref{thm:qtopm}, application of Schaeffer's prescription leads to a planar map.
This planar map is naturally labeled, since its vertices are also vertices of $Q$, and it can be rooted at the corner containing the end point with largest label of the root edge of $Q$.
Hence, one obtains a labeled rooted planar map, which we denote by $\Psi(Q)\in\mathcal{M}^{(l)}$.

\begin{figure}[t]
\centering
\includegraphics[width=4.5cm]{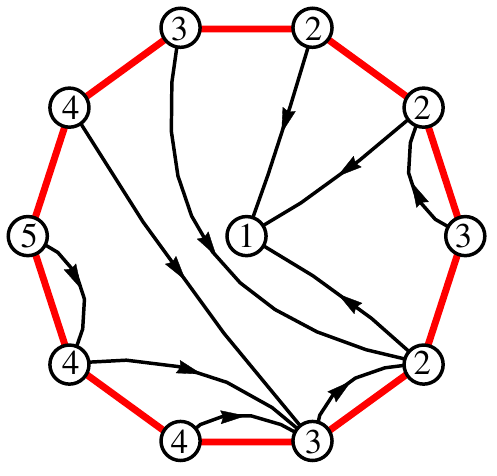}
\caption{An example of the operation performed on a face of a labeled planar map in the definition of $\Xi$.}%
\label{fig:inversebij}
\end{figure}

\begin{theorem}\label{thm:genbijection}
Schaeffer's prescription defines a bijection $\Psi : \mathcal{Q}^{(l)} \to \mathcal{M}^{(l)}$ with the following properties:
\begin{enumerate}
\item The number of edges of $\Psi(Q)$ equals the number of faces $N$ of $Q$.
\item The number of faces of $\Psi(Q)$ equals the number of local minima $N_{\text{min}}$ of $Q$.
\item The number of local maxima of $\Psi(Q)$ equals the number of local maxima $N_{\text{max}}$ of $Q$.
\item The maximal label of $\Psi(Q)$ equals the maximal label of $Q$.
\item The minimal label of $\Psi(Q)$ equals the minimal label of $Q$ plus one.
\item The label on the root of $\Psi(Q)$ equals the maximal label on the root edge of $Q$.
\end{enumerate}
\end{theorem}

\begin{proof}
Let us start by verifying properties (i)-(vi) for a quadrangulation $Q$.
Properties (i) and (ii) have already been proven in lemma 1.
To show (iii), let $v$ be a vertex of $Q$ with label $t$.
Suppose that $v$ is a local maximum on the quadrangulation, meaning that all its neighbors have smaller label, then it has maximal label in each of the adjacent faces.
By inspection of the rules in figure \ref{fig:quad0}(a) there is an edge of $\Psi(Q)$ leading away from $v$ ending at a vertex with smaller or equal label for each adjacent face in $Q$.
Therefore $v$ is also a local maximum in $\Psi(Q)$.
If $v$ is not a local maximum, it is either a local minimum, in which case it does not appear as a vertex on $\Psi(Q)$, or both $t+1$ and $t-1$ appear as labels on its neighbors in $Q$.
In the latter case there must be a simple face adjacent to $v$ for which the colouring leads to an edge connecting $v$ to a vertex with label $t+1$.
Then $v$ is not a local maximum of $\Psi(Q)$.
Hence, the local maxima of $Q$ and $\Psi(Q)$ coincide, which completes the proof of (iii).
Moreover, (iv) follows from the fact that the vertices with maximal label are necessarily local maxima.
Properties (v) is a consequence of the existence of at least one edge leading away from a vertex with minimal label, the end point of which is not a local minimum and therefore a vertex of $\Psi(Q)$.
Finally, (vi) is a direct consequence of the construction.

To show that $\Psi$ is a bijection, let us introduce a map $\Xi : \mathcal{M}^{(l)} \to \mathcal{Q}^{(l)}$, which we will later show to be the inverse of $\Psi$.
Let $M\in\mathcal{M}^{(l)}$ be a labeled rooted planar map.
For each face $\mathcal{F}$ of $M$ the following operation is performed.
Let $t_{\mathcal{F}}$ be the minimal label of the corners of $\mathcal{F}$.
A new vertex is inserted in the interior of $\mathcal{F}$ with label $t_{\mathcal{F}}-1$.
Then for each corner a new edge is drawn.
If the corner is labeled $t_{\mathcal{F}}$ it is connected to the new vertex, otherwise it is connected to the first corner with smaller label that one encounters when traversing the boundary of $\mathcal{F}$ in anti-clockwise order.
See figure \ref{fig:inversebij} for an example.
Up to continuous deformations a unique way of drawing the new edges without crossing exists.
To show this, it suffices to consider the edges leaving the corners with label larger than $t_{\mathcal{F}}$.
Let $e$ be such an edge with labels $t$ and $t-1$ on its end points, then by construction the corners on the right-hand side of $e$ have label larger or equal to $t$.
Therefore they will never appear as end point of an edge starting at a corner on the left-hand side of $e$.
Hence, $e$ can be drawn uniquely without crossing any other edges by keeping on its left-hand side both the new vertex and the previously drawn edges that start at corners on its left-hand side.

We claim that application of this operation to each face followed by deletion of the original edges gives a quadrangulation of the sphere.
In order to show this let us consider the combined planar map $M'$ containing both the old and new edges.
Each face $\mathcal{F}'$ of $M'$ is by construction bounded by exactly one old edge $e$.
Let us orient $e$ such that $\mathcal{F}'$ is on its right-hand side and denote the labels of the starting and end point of $e$ by $t_0$ and $t_1$.
As a result of the drawing rules, if $t_1-t_0=1$ the degree of $\mathcal{F}'$ is two, since there is a new edge with the same end points as $e$.
If $t_1-t_0=0$ the degree is three, because both end points of $e$ are connected by new edges to the same vertex with label $t_0-1$.
Finally, it can be seen that when $t_1-t_0=-1$ the degree of $\mathcal{F}'$ is four.
Indeed, the end point of $e$ is connected to the first corner $c_1$ with label $t_1-1$ encountered when running anti-clockwise around the face.
But this is also the first corner with label $t_1-1$ after the first corner $c_0$ with label $t_1=t_0-1$, to which the starting point of $e$ is connected.
By inspecting the degrees of the faces on either side of the edge $e$, we conclude that deletion of $e$ in any case results in a face of degree four.
Hence, removing all old edges from $M'$ gives a quadrangulation $\Xi(M)$, which is labeled by construction and can be rooted on the edge that was drawn from the root corner of $M$.

With the explicit construction of $\Xi$ in hand it is straightforward to show that $\Psi\circ\Xi$ is the identity on $\mathcal{M}^{(l)}$.
Indeed, as we have seen above, each face $\mathcal{F}$ of $\Xi(M)$ is formed by merging the faces on either side of an edge $e$ of $M$ with labels $t_0$ and $t_1$.
If $t_0=t_1$ the face $\mathcal{F}$ is confluent and $e$ is the diagonal connecting the corners of $\mathcal{F}$ with largest label.
Otherwise $\mathcal{F}$ is simple and $e$ connects the corner with largest label to the next corner in clockwise direction around the face.
In both cases these are precisely the coloured edges that are to be drawn according to Schaeffer's prescription.
Hence, $\Psi(\Xi(M))$ has the same edges as $M$, while the labels on $M$ are unaffected by the operation and $\Psi(\Xi(M))$ is easily seen to be rooted at the root corner of $M$.

Finally let us prove that $\Xi\circ\Psi$ is the identity on $\mathcal{Q}^{(l)}$.
To do this we show that all of the edges of a quadrangulation $Q\in\mathcal{Q}^{(l)}$ arise from the operation described above on the faces of the planar map $M =\Psi(Q)$.
First of all, let $e$ be an edge of $Q$ that ends on a local minimum of $Q$.
Since the starting point of $e$ is not a local minimum, it corresponds to a corner of the face $\mathcal{F}$ of $M$ that contains the end point of $e$ in its interior.
Since this corner has label exactly one larger than the vertex in the interior of $\mathcal{F}$, the edge $e$ is indeed created by the operation.
Now let $e$ be an edge for which neither end point is a local minimum.
We orient the edge such that its labels are of the form $(t,t-1)$.
Since both end points of $e$ are on $M$, $e$ defines a chord of a face $\mathcal{F}$ of $M$.
We claim that the corners of $\mathcal{F}$ that lie on the right-hand side of $e$ have label larger or equal to $t$, which implies that the edge $e$ is created by the operation of $\Xi$.
To see this we inspect the quadrangle directly on the right-hand side of $e$.
As shown in figure \ref{fig:invbijproof}, three situations are possible, but in each case the corners of the quadrangle that correspond to new corners of $\mathcal{F}$ have label larger or equal to $t$.
Moreover, the other edges of the quadrangle that lie in the same face $\mathcal{F}$ are again chords with labels larger or equal to those of $e$.
The argument can be repeated with the new chords showing that no corners with label smaller than $t$ can appear on the right-hand side of $e$.
Hence, all edges of $Q$ are created in the operation of $\Xi$ on the faces of $M$ and since the number of corners of $M$ is equal to the number of edges of $Q$, $Q$ can have no additional edges.
\end{proof}

\begin{figure}[t]
\centering
\includegraphics[width=10.5cm]{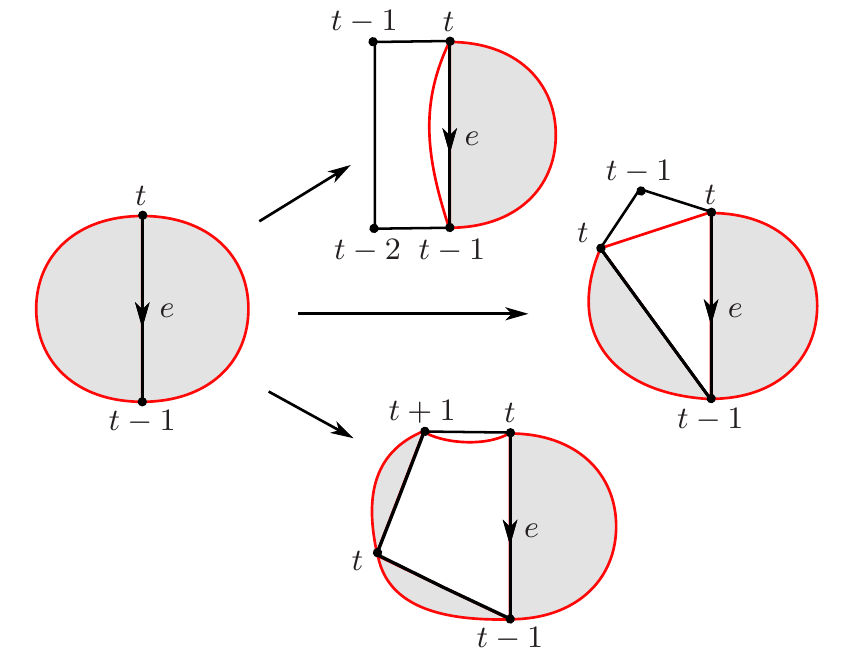}
\caption{Illustration of the fact that the corners to the right of $e$ have label larger than $t-1$. The red curve represents the boundary of the face $\mathcal{F}$, while the shaded areas represents the unknown tiling of the $\mathcal{F}$ by faces of $Q$.}%
\label{fig:invbijproof}
\end{figure}

The bijection $\Psi$ of theorem \ref{thm:genbijection} is quite useful as it gives rise to a whole bunch of specialized bijections when restricted to subsets of $\mathcal{Q}^{(l)}$ defined by certain restrictions on the properties (i)-(vi).
For instance, the trivial bijection is obtained from theorem \ref{thm:genbijection} by restricting the labels on the quadrangulations to take the values $0$ or $1$.
To retrieve the Cori--Vauquelin--Schaeffer bijection from section \ref{sec:schaeffer} we need the following simple lemma.

\begin{lemma}\label{thm:dist}
Let $M\in\mathcal{M}^{(l)}$ be a labeled rooted planar map.
If $M$ has a single local minimum labeled zero, then the labeling corresponds to the distance labeling to the vertex labeled zero.
The same holds for labeled rooted quadrangulations, since $\mathcal{Q}^{(l)}\subset\mathcal{M}^{(l)}$.
\end{lemma}
\begin{proof}
Each vertex of $M$ with label $t>0$ is not a local minimum and therefore has a ``descending'' edge connecting it to a vertex of label $t-1$. 
Starting at a vertex with label $t>0$ one can find a path of length $t$ to the vertex with label $0$ by repeatedly traversing descending edges. 
Since the labels differ by at most one along the edges, no shorter path exists.
\end{proof}

It follows from lemma \ref{thm:dist} that the subset of $\mathcal{Q}^{(l)}$ given by quadrangulations $Q$ with a single local minimum labeled zero is in bijection with the set of rooted pointed quadrangulations.
According to theorem 1 it is also in bijection with the set of labeled rooted planar maps with one face and minimal label one, i.e. the set of rooted well-labeled trees.
Moreover, property (iii) justifies the counting of the local maxima of $\Psi(Q)$ instead of the local maxima of $Q$ used in section \ref{sec:gencdt}.

Instead of requiring a single local minimum one can also require a single local maximum labeled zero.
In that case the labels on the quadrangulations $Q$ represent minus the distance to the local maximum, again giving a subset of $\mathcal{Q}^{(l)}$ in bijection with the set of rooted pointed quadrangulations.
But now according to theorem 1 this subset is in bijection with the set of labeled rooted planar maps $M$ with a single local maximum with label 0.
It follows, again from lemma \ref{thm:dist}, that the labels on $M$ also correspond to minus the distance to the local maximum and are fixed by the choice of a distinguished vertex with label zero.
We therefore have the following result, where we flipped the sign of all the labels.

\begin{theorem}\label{thm:quadmap}
Application of the prescription in figure \ref{fig:quadtomap}a to the distance labeling of a rooted pointed quadrangulation $Q$ gives a planar map $\Phi(Q)$ with a distinguished vertex, which can be rooted at the corner containing the end point with smallest label of the root edge of $Q$.
The map $\Phi : \mathcal{Q} \to \mathcal{M}$ is a bijection between the set $\mathcal{Q}$ of rooted pointed quadrangulations of the sphere and the set $\mathcal{M}$ of rooted pointed planar maps.
It satisfies the following properties:
\begin{enumerate}
\item The number of edges of $\Phi(Q)$ equals the number of faces of $Q$.
\item The number of faces of $\Phi(Q)$ equals the number of local maxima $N_{\text{max}}$ of $Q$.
\item The distance of the origin to the root vertex in $\Phi(Q)$ equals the distance of the origin to the closest end point of the root edge in $Q$.
\end{enumerate}
\end{theorem}

\begin{proof} 
See discussion above. Properties (i), (ii) and (iii) follow directly from properties (i), (ii) and (vi), respectively, of theorem 1.
\end{proof}

In figure \ref{fig:quadtomap}c we have shown the result of the new colouring rules for the same quadrangulation as in figure \ref{fig:quad0}.

\begin{figure}[t]
{\centering
\includegraphics[width=\linewidth]{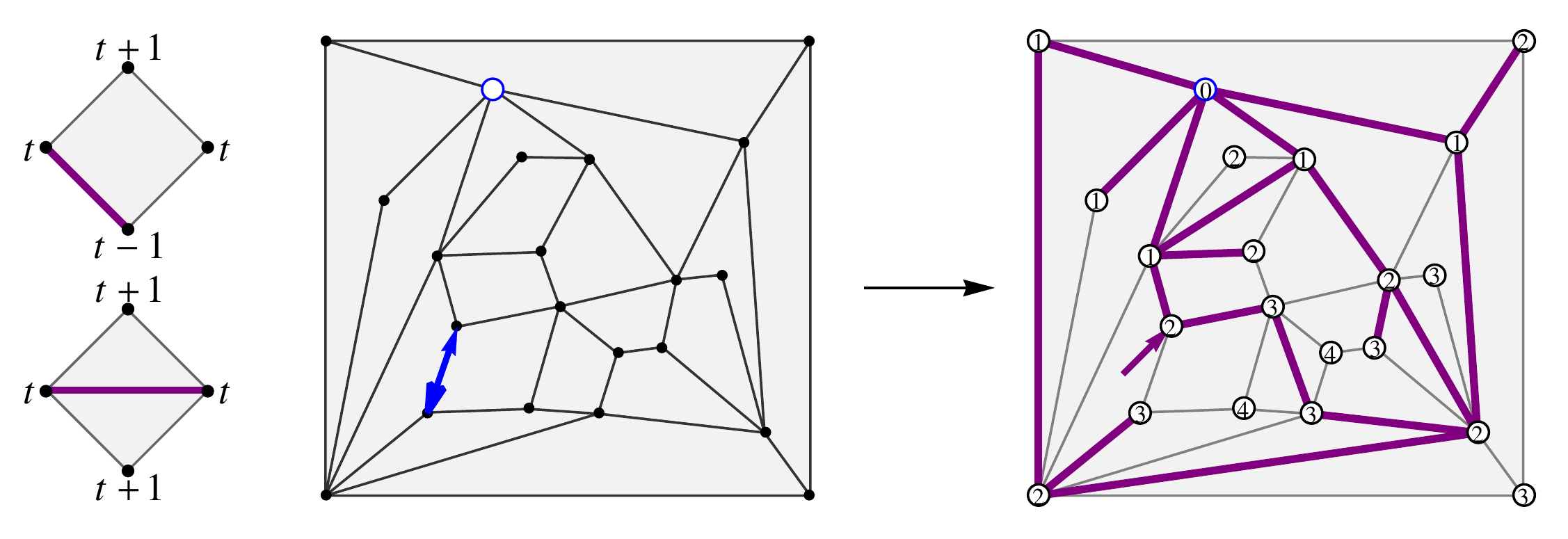}}
\vspace{-0.6cm}

\hspace{1cm}(a)\hspace{3.9cm}(b)\hspace{6.6cm}(c)
\caption{The prescription (a) opposite to Schaeffer's prescription gives a bijection between rooted pointed quadrangulations (b) and rooted pointed planar maps (c). }%
\label{fig:quadtomap}
\end{figure}

As a consequence of theorem \ref{thm:quadmap} the counting of rooted pointed quadrangulations with a fixed number of faces $N$ and a fixed number of local maxima $N_{\text{max}}$ coincides with the counting of rooted pointed planar maps with $N$ edges and $N_{\text{max}}$ faces.\footnote{Let us remark that combining $\Phi$ with the trivial bijection gives a bijection $\Phi_0^{-1} \circ \Phi$ of the set of rooted pointed quadrangulations onto itself. A quadrangulation with $N_{\text{max}}$ local maxima is mapped to a quadrangulation with precisely $N_{\text{max}}$ vertices at odd distance of the origin, which therefore also obey the same counting.}
In fact, we may conclude that the generalized CDT amplitudes appear as the natural scaling limit of random planar maps in which the number of faces is kept finite.

\begin{figure}[t]
\centering
\begin{tabular}{|r|r|rrrrrr|}
\hline & & $t=0$ & $t=1$ & $t=2$ & $t=3$ & $t=4$ & $t=5$ \\
\hline $N=1$ & $n=1$ & 1 & 1 &  &  &  &  \\
& $n=2$ & 1 &  &  &  &  &  \\
\hline $N=2$ & $n=1$ & 2 & 3 & 1 &  &  &  \\
& $n=2$ & 5 & 5 &  &  &  &  \\
& $n=3$ & 2 &  &  &  &  &  \\
\hline $N=3$ & $n=1$ & 5 & 9 & 5 & 1 &  &  \\
& $n=2$ & 22 & 34 & 10 &  &  &  \\
& $n=3$ & 22 & 22 &  &  &  &  \\
& $n=4$ & 5 &  &  &  &  &  \\
\hline $N=4$ & $n=1$ & 14 & 28 & 20 & 7 & 1 &  \\
& $n=2$ & 93 & 175 & 89 & 15 &  &  \\
& $n=3$ & 164 & 258 & 70 &  &  &  \\
& $n=4$ & 93 & 93 &  &  &  &  \\
& $n=5$ & 14 &  &  &  &  &  \\
\hline $N=5$ & $n=1$ & 42 & 90 & 75 & 35 & 9 & 1 \\
& $n=2$ & 386 & 813 & 546 & 165 & 20 &  \\
& $n=3$ & 1030 & 1993 & 954 & 143 &  &  \\
& $n=4$ & 1030 & 1640 & 420 &  &  &  \\
& $n=5$ & 386 & 386 &  &  &  &  \\
& $n=6$ & 42 &  &  &  &  &  \\
\hline
\end{tabular}
\caption{The first few non-zero values of the number $\mathcal{N}_t(N,n)$ of rooted pointed maps with $N$ edges, $n$ faces and the origin at distance $t$ from the root.}%
\label{fig:table}
\end{figure}


In section \ref{sec:timedep} we have studied the distribution of distances between the root and the origin in random rooted pointed quadrangulation with a certain number of local maxima.
As a byproduct of this analysis we obtain expressions for the distribution of distances in random planar maps as a function of the number of edges and the number of faces.

Recall from section \ref{sec:timedep} that $z_0(t)-z_0(t-1)$, with $z_0(t)$ as in \eqref{eq:ztsolution}, is the generating function for the number of rooted pointed quadrangulations with $N$ faces and $n$ local maxima for which the origin is a distance $t$ from the furthest end of the root edge.
Hence, we find the generating function
\begin{equation}
z_t(g,\g) := z_0(t+1)-z_0(t) = \sum_{N=0}^{\infty} \sum_{n=0}^{N+1} \mathcal{N}_t(N,n) g^N\, \g^n
\end{equation}
for the number $\mathcal{N}_t(N,n)$ of rooted pointed quadrangulations with the origin exactly at distance $t$ from the \emph{closest} end of the root edge.
According to theorem \ref{thm:quadmap}, property (iii), these quadrangulations are mapped by $\Phi$ exactly onto rooted pointed planar maps with $N$ edges, $n$ faces and the origin at distance $t$ from the root, which are therefore also counted by $z_t(g,\g)$.
It is a direct generalization of the generating function $z_{t=0}(g,\g)$ for rooted planar maps with $N$ edges and $n$ faces first derived by Tutte in \cite{tutte_enumeration_1968}.
On the other hand it generalizes the two-point functions for planar maps derived in \cite{bouttier_geodesic_2003,di_francesco_geodesic_2005,bouttier_planar_2012} of which the simplest versions roughly correspond to $z_{t}(g,\g=1)$.
By plugging the solutions to (\ref{eq:zrecur}) and (\ref{eq:betasigma}) into (\ref{eq:ztsolution}) one can explicitly compute the coefficients $\mathcal{N}_t(N,n)$.
The first few non-zero values of $\mathcal{N}_t(N,n)$ are shown in table \ref{fig:table} and are checked to agree with a brute-force enumeration of all planar maps up to $N=5$.

\section{Loop identities}\label{sec:loops}

A process was studied in \cite{ambjorn_putting_2007,ambjorn_string_2008} in which two universes merge and the resulting universe disappears into the vacuum.
From a two-dimensional perspective, we consider the amplitude for surfaces with two boundary loops of length $L_1$ and $L_2$ respectively separated by a geodesic distance $D$ (figure \ref{fig:twoloop}).
Contrary to the case of the propagator, where all points on the final boundary are required to have a fixed distance to the initial boundary, here we only fix the minimal distance between the boundaries.
To get a generalized CDT amplitude, one has to specify the time on the boundaries, which is taken to be constant $T_1$ and $T_2$ respectively.
In order to get a continuous time function throughout the surface, the boundary times $T_1$ and $T_2$ and the distance $D$ have to satisfy the inequality $|T_1 - T_2| \leq D$.
The resulting amplitude is denoted by $G_{\lambda,\g_s}(L_1,L_2;T_1,T_2;D)$.

In \cite{ambjorn_putting_2007,ambjorn_string_2008} it was shown that $G_{\lambda,\g_s}(L_1,L_2;T_1,T_2;D)$ can be expressed in terms of two propagators and a cap function.
A non-trivial calculation showed that, quite remarkably, the amplitude $G_{\lambda,\g_s}(L_1,L_2;T_1,T_2;D)$ does not depend on $T_1$ and $T_2$ at all,
\begin{equation}\label{eq:twoloopidentity}
G_{\lambda,\g_s}(L_1,L_2;T_1,T_2;D) = G_{\lambda,\g_s}(L_1,L_2;D) \quad (\text{provided }|T_1 - T_2| \leq D).
\end{equation}
Even though the same ensemble of surfaces contributes to the loop-loop amplitudes with different initial times, like the ones in figure \ref{fig:twoloop}, this is non-trivial because the weight assigned to each surface depends on the number of local maxima of the time functions and therefore on the boundary times.

\begin{figure}[t]
\centering
\includegraphics[height=4.5cm]{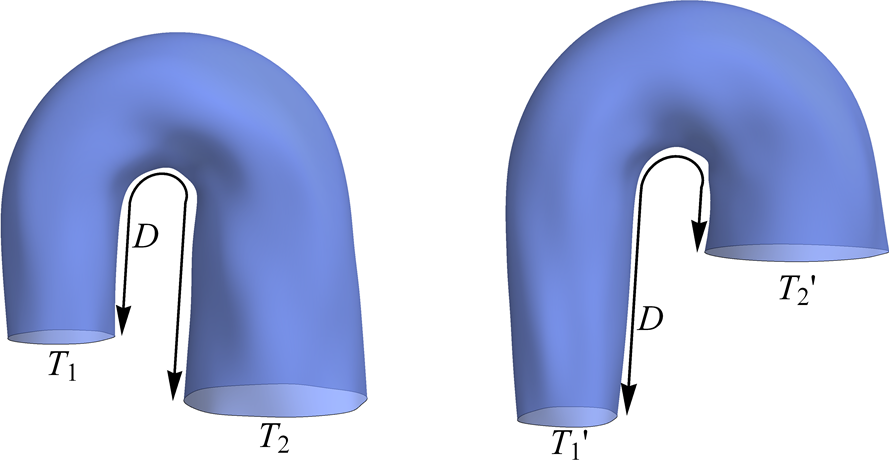}

(a)\hspace{4.5cm}(b)
\caption{The two-loop amplitudes $G_{\lambda,\g_s}(L_1,L_2;T_1,T_2;D)$ .}%
\label{fig:twoloop}
\end{figure}

A similar identity can be derived for Euclidean two-dimensional gravity \cite{aoki_operator_1996}, in which case it can be understood as a consistency relation for the continuum amplitudes.
From the discrete point of view, e.g. in terms of the quadrangulations discussed previously, it is clear that such an identity must hold since for $\g=1$ the contribution of each quadrangulations is independent of the labeling.
In the following we will show that even for $\g \neq 1$ the discrete two-loop amplitude $\mathcal{G}_{\g}(l_1,l_2;t_1,t_2;d)$, which is the discrete analogue of $G_{\lambda,\g_s}(L_1,L_2;T_1,T_2;D)$, is independent of $t_1$ and $t_2$. 
In fact, we will show that for any two pairs $(t_1,t_2)$ and $(t_1',t_2')$ there exists a bijection from the set of quadrangulations with two boundaries onto itself, such that a quadrangulation with a number of local maxima with respect to the pair $(t_1,t_2)$ is mapped to a quadrangulation with the same number of local maxima with respect to the pair $(t_1',t_2')$.
The continuum identity (\ref{eq:twoloopidentity}) and its discrete analogue are a direct consequence of the existence of such a bijection.

Before we consider quadrangulations with two boundaries, let us consider the set $\mathcal{Q}_d$ of rooted quadrangulations $Q$ of the sphere with two marked vertices $v_1$ and $v_2$ separated by a distance $d$ along the edges of $Q$.
Given a pair of integers $(t_1,t_2)$ satisfying $|t_1-t_2| < d$ and $t_1+t_2+d$ even, one can label the vertices of the quadrangulation in a unique way such that the only local minima occur at the vertices $v_i$, which are labeled $t_i$, and the labels vary by exactly one along the edges.
The labeling gives the distance to $v_1$ shifted by $t_1$ or the distance to $v_2$ shifted by $t_2$, depending on which one is smaller (see \cite{miermont_tessellations_2009} for a similar construction).
According to (the sign-flipped version of) theorem \ref{thm:genbijection} the rooted labeled quadrangulation is mapped by $\Psi$ to a rooted labeled planar map for which the only two local minimal are again $v_i$ with label $t_i$.
As for the quadrangulation the labeling on the planar map is the distance to $v_1$ shifted by $t_1$ or the distance to $v_2$ shifted by $t_2$, depending on which one is smaller.
Hence, one obtains a rooted planar map $\Phi^d_{t_1,t_2}(Q)$ with two marked vertices, of which the labeling is canonical.

\begin{theorem}
For any pair of integers $(t_1,t_2)$ satisfying $d-|t_1-t_2|\in\{2,4,6,\ldots\}$ the map
$\Phi^d_{t_1,t_2}: \mathcal{Q}_d \to \mathcal{M}_d \cup \mathcal{M}_{d-1}$,
where $\mathcal{M}_d$ is the set of rooted planar maps with two marked vertices separated by a distance $d$, is a bijection.
\end{theorem}

\begin{proof}
It follows from (the sign-flipped version of) theorem \ref{thm:genbijection} that the set $\mathcal{Q}_{t_1,t_2}$ of rooted labeled quadrangulations with two local minima labeled $t_1$ and $t_2$ is in bijection with the set $\mathcal{M}_{t_1,t_2}$ of rooted labeled planar maps again with two local minima labeled $t_1$ and $t_2$.
Any $Q\in\mathcal{Q}_{t_1,t_2}$ has a distance $d$ between the local minima satisfying $d = |t_1-t_2| + 2 k$ for some $k \geq 1$, because of the bipartiteness of $Q$. 
This means that we have a natural bijection 
\begin{equation}\label{eq:tbij0}
\mathcal{Q}_{t_1,t_2} \longleftrightarrow \bigcup_{k=1}^{\infty} \mathcal{Q}_{d=|t_1-t_2| + 2 k},
\end{equation}
given by the canonical labeling discussed above.
Likewise, any $M\in\mathcal{M}_{t_1,t_2}$ has a distance $d = |t_1-t_2|+m$ for some $m\geq 1$ between its local minima.
Therefore we have a natural bijection
\begin{equation}\label{eq:tbij1}
\mathcal{M}_{t_1,t_2} \longleftrightarrow \bigcup_{m=1}^{\infty} \mathcal{M}_{d=|t_1-t_2| + m} = \bigcup_{k=1}^{\infty} (\mathcal{M}_{d=|t_1-t_2| + 2k-1}\cup \mathcal{M}_{d=|t_1-t_2| + 2k}).
\end{equation}
To show that $\Phi^d_{t_1,t_2}$ defines a bijection between the individual sets $\mathcal{Q}_{d=|t_1-t_2| + 2 k}$ and $\mathcal{M}_{d=|t_1-t_2| + 2k-1}\cup \mathcal{M}_{d=|t_1-t_2| + 2k}$ appearing on the right-hand side of (\ref{eq:tbij0}) and (\ref{eq:tbij1}) it is sufficient to check that $\Phi^d_{t_1,t_2}(\mathcal{Q}_d) \subset \mathcal{M}_d \cup \mathcal{M}_{d-1}$.

Given a rooted quadrangulation $Q\in\mathcal{Q}_d$ with two marked vertices separated by a distance $d$ (see figure \ref{fig:twopointquad}a for an example with $d=4$).
We will show that $v_1$ and $v_2$ are separated by $d$ or $d-1$ edges in the planar map $M=\Phi^d_{t_1,t_2}(Q)$.
For convenience we assign a \emph{type} to vertices $v$ in the quadrangulation $Q$ according to their distances $d(v,v_i)$ to $v_1$ and $v_2$.
A vertex $v$ is of type 1 if $d(v,v_1)+t_1 < d(v,v_2)+t_2$, of type 2 if $d(v,v_1)+t_1 > d(v,v_2)+t_2$, and of type 0 in case of equality (see figure \ref{fig:twopointquad}b).
For $i=1,2$ a type-$i$ vertex $v$ labeled $t$ that is not a local maximum has a distance $t-t_i$ in the planar map to $v_i$.
Due to triangle inequalities all vertices labeled $t < t_{\mathrm{max}}:=(t_1+t_2+d)/2$ are either of type 1 or type 2.
Any path connecting $v_1$ and $v_2$ in the planar map must include a vertex labeled $t_{\mathrm{max}}-1$ of type 1 and one of type 2, and therefore its length is at least $d-1$.
Since the distance between $v_1$ and $v_2$ in the quadrangulation is $d$, there must exist at least one vertex $v_{\mathrm{max}}$ with $d(v_{\mathrm{max}},v_1)+t_1=d(v_{\mathrm{max}},v_2)+t_2=t_{\mathrm{max}}$, which is the vertex of maximal label on a geodesic connecting $v_1$ and $v_2$ (see figure \ref{fig:twopointquad}b).
Let us consider the cycle of neighbours of $v_{\mathrm{max}}$ in $Q$ in anticlockwise order.
They come in three types: vertices of type $0$ labeled $t_{\mathrm{max}}+1$, vertices of type 1, respectively of type 2, labeled $t_{\mathrm{max}}-1$.
Moreover, at least one vertex of both type 1 and type 2 must occur.
Now there are two possibilities for the cycle: either a type-1 vertex is adjacent to a type-2 vertex, or both a type-1 vertex and a type-2 vertex are followed by a type-0 vertex (examples of both situations are indicated by $v_{\text{max}}$ and $v_{\text{max}}'$ in figure \ref{fig:twopointquad}b).
In the first case the type-1 vertex and the type-2 vertex are opposite corners of a confluent face and are therefore connected by an edge in the planar map.
Hence, there is a path of length $d-1$ connecting $v_1$ and $v_2$.
In the second case $v_{\mathrm{max}}$ is connected by an edge in the planar map to both a type-1 vertex and a type-2 vertex, resulting in a path of length $d$.
Hence, $\Phi^d_{t_1,t_2}(\mathcal{Q}_d) \subset \mathcal{M}_d \cup \mathcal{M}_{d-1}$.
\end{proof}

\begin{figure}[t]
\centering
\includegraphics[height=6.2cm]{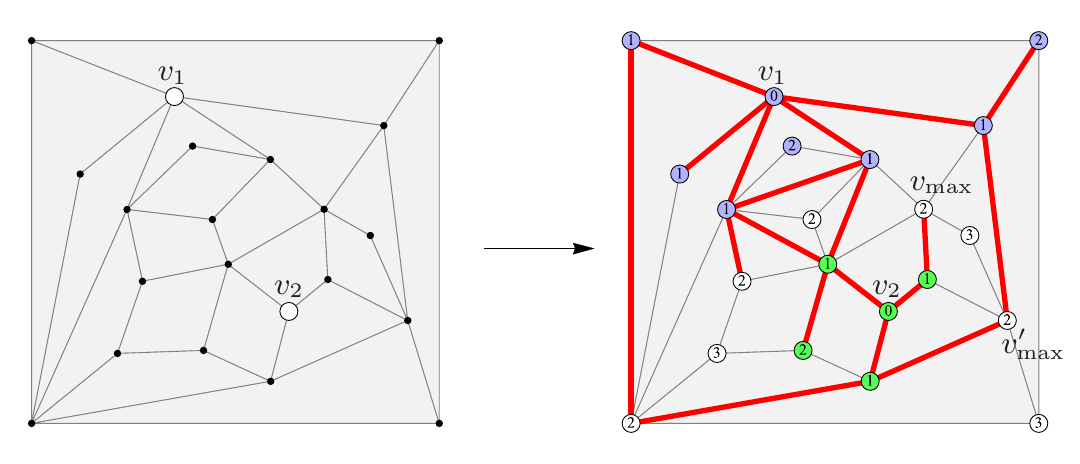}

(a)\hspace{7.4cm}(b)
\caption{(a) An example of a quadrangulation $Q$ with two marked vertices separated by a distance $d=4$. 
(b) The planar map resulting from the canonical labeling with $t_1=t_2=0$.
 The vertices of type 0, 1 and 2 are coloured white, blue and green, respectively.
 Both the vertices $v_{\text{max}}$ and $v_{\text{max}}'$ are of type 0 and have label equal to $t_{\mathrm{max}}=2$ (notice that there are three more vertices satisfying the same property).
 The vertex $v_{\text{max}}$ is adjacent to a face containing both a type-1 vertex and a type-2 vertex which are connected in the planar map by the diagonal, hence giving rise to a path of length $d-1$ connecting $v_1$ and $v_2$.
 The vertex $v_{\text{max}}'$ on the other hand has a type-1 and type-2 vertex both followed by a type-0 vertex in its anticlockwise sequence of neighbors.
 Therefore there is a path of length $d$ passing via $v_{\text{max}}'$.}%
\label{fig:twopointquad}
\end{figure}

From the bijection $\Phi^d_{t_1,t_2}$ we can easily construct the bijection
\begin{equation}\label{eq:bijectiont1t2}
{\Phi_{t_1',t_2'}^d}^{-1} \circ \Phi_{t_1,t_2}^d : \mathcal{Q}_d(N) \to \mathcal{Q}_d(N), \quad |t_1-t_2|,|t_1'-t_2'| \in \{d-2,d-4,\ldots\},
\end{equation}
which maps the set of rooted quadrangulations with two marked points separated by a distance $d$ to itself.
A quadrangulation with a number $N_{\mathrm{max}}$ of local maxima of the canonical labeling w.r.t. $(t_1,t_2)$ is mapped to a quadrangulation with $N_{\mathrm{max}}$ local maxima of the canonical labeling w.r.t. $(t_1',t_2')$.

The bijection (\ref{eq:bijectiont1t2}) can be extended to quadrangulations with two boundaries instead of two marked points by gluing disks to the boundaries as we did for the propagator in section \ref{sec:timedep}.
We consider quadrangulations $Q$ with two boundaries, one of length $2l_1$ labeled alternatingly by $t_1$ and $t_1+1$, and another of length $2l_2$ labeled $t_2$ and $t_2+1$.
We fix the smallest distance between the points labeled $t_1$ on the first boundary and the points labeled $t_2$ on the second boundary to be $d$, subject to the same inequalities as above, i.e. $d-|t_1-t_2|$ positive and even.
By gluing disks constructed from $l_i$ simple faces to the boundaries, a quadrangulation is obtained with two marked vertices $v_i$ labeled $t_i-1$ and separated by a distance $d+2$.
The bijection (\ref{eq:bijectiont1t2}) leaves invariant the structure of the quadrangulation in the direct neighbourhood of $v_i$, and therefore the disks can be removed again after the bijection to obtain a quadrangulation with different labels $(t_1',t_2')$ on its boundaries.

It follows immediately that the discrete two-loop amplitude $\mathcal{G}_{\g}(l_1,l_2;t_1,t_2;d)$, i.e. the generating function for such quadrangulations including a factor of $\g$ for each local maximum of the canonical labeling, is independent of $t_1$ and $t_2$ (as long as $d-|t_1-t_2|$ is positive and even).

\section{Triangulations}\label{sec:triangulations}

For the sake of completeness we will show that most constructions in this paper can also be carried out for triangulations.
As we will see, the analogues of the bijections described in section \ref{sec:bijections} are not as simple for triangulations but still manageable.
The bijection we will use is a special case of the Bouttier--Di Francesco--Guitter bijection between arbitrary planar maps and \emph{labeled mobiles} introduced in \cite{bouttier_planar_2004}.
In a slightly different formulation it was used by Le Gall in \cite{gall_uniqueness_2011} to prove that random triangulations and random quadrangulations as metric spaces converge in a quite general way to the same continuum object, known as the \emph{Brownian map}.

\begin{figure}[t]
\centering
\subfloat[]{
\centering
\includegraphics[height=5cm]{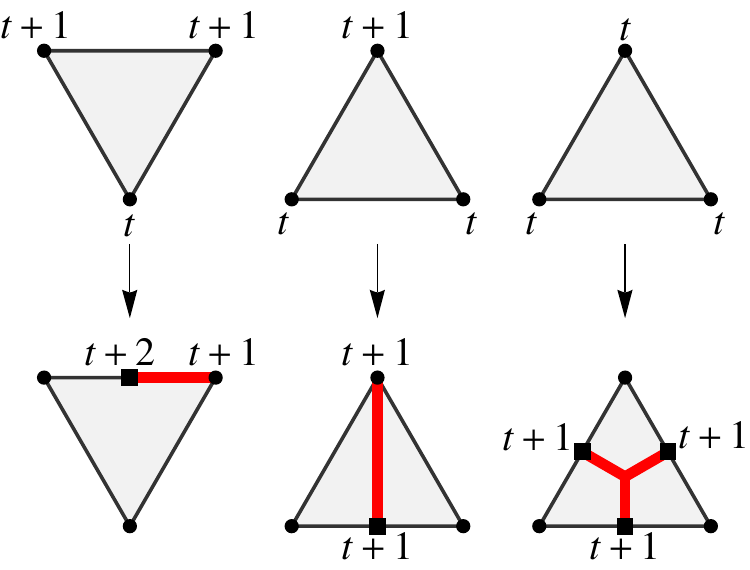}
\label{fig:trirules}
}
\subfloat[]{
\centering
\includegraphics[height=5.5cm]{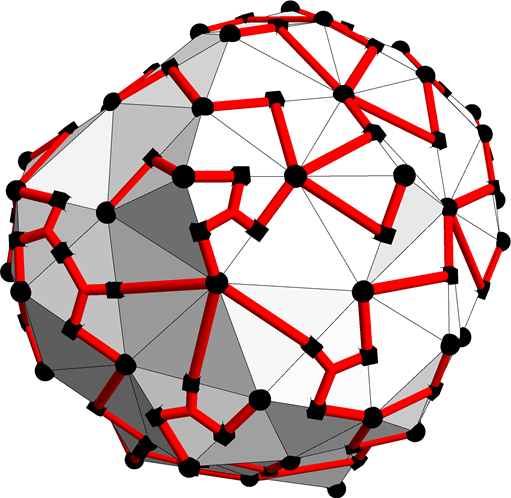}
\label{fig:tribijection}
}
\caption{(a) The rules for a ``down'' triangle, an ``up'' triangle, and a ``flat'' triangle respectively. (b) An example of the application of the rules to a triangulation of the sphere.}%
\label{fig:tribijections}
\end{figure}

The bijection for triangulations can be understood as a special case of the Cori--Vauquelin--Schaeffer bijection for quadrangulations. 
Given a triangulation of the sphere with $N$ triangles and one vertex marked as the origin, one can label all the vertices according to their distance to the origin along the edges.
A new vertex labeled $t+1$ is inserted in the middle of each edge connecting vertices of equal label $t$.
Each triangle with equal labels $t$, called a \emph{flat} triangle, is subdivided into three faces by connecting each of the three new vertices on its edges to a new vertex labeled $t+2$ in the center of the triangle.
The resulting planar map is a quadrangulation labeled by the distance to the origin.
It is convenient to keep track of the vertices belong to the triangulation (type 1), the vertices lying on the edges of the triangulation (type 2), and the vertices in the centers of the flat triangles (type 3).
The first two types are depicted in figure \ref{fig:tribijections} by disks and squares respectively, while the type-3 vertices correspond to the unmarked intersections.

The application of Schaeffer's prescription (figure \ref{fig:quad0}a) to the resulting quadrangulation turns out to be equivalent to the prescription in figure \ref{fig:trirules} for the triangles. 
Not all quadrangulations arise from triangulations, hence a limited class of labeled trees will appear.
Before discussing this class, it is convenient to switch to the rooted versions of the triangulations and trees, like in section \ref{sec:schaeffer}.
A triangulation is rooted by distinguishing an edge, which for simplicity we demand to connect vertices of different label.\footnote{See \cite{gall_uniqueness_2011} for the general case where any oriented edge can be used as root.}
Since this edge appears in the quadrangulation with a vertex of type 1 at its end point with largest label, the tree is naturally rooted at this vertex.
From the way the triangles in figure \ref{fig:trirules} can be glued, it can be seen that the class of rooted, labeled trees satisfies the following rules (see \cite{gall_uniqueness_2011}). 
The root is of type 1.
A vertex of type 1 labeled $t$ has zero or more children of type 2 labeled $t$ or $t+1$.
A vertex of type 2 labeled $t$ has either one child of type 1 labeled $t$ or $t-1$, or two children of type 2 labeled $t$.
Notice that in the latter case we regard the two type-2 vertices of a flat triangle to be directly connected to the other type-2 vertex, instead of keeping track of the type-3 vertex in between.  

\begin{figure}[t]
\centering
\subfloat[]{
\centering
\includegraphics[width=4.5cm]{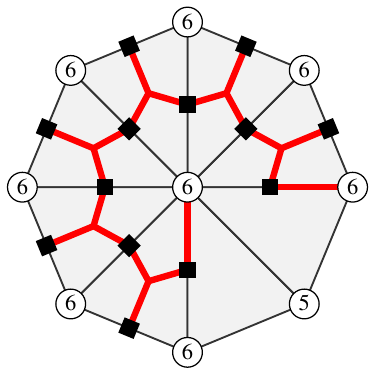}
\label{fig:trimax1}
}
\subfloat[]{
\centering
\includegraphics[width=4.5cm]{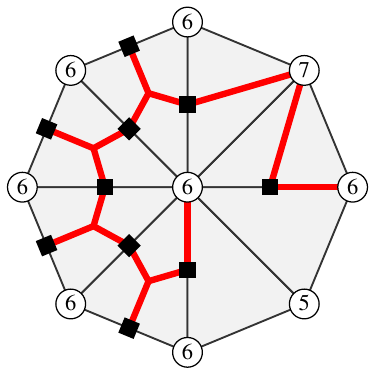}
\label{fig:trimax2}
}
\caption{An example of a vertex that is (a) a local maximum and (b) one that is not.}%
\label{fig:trimax}
\end{figure}

One would like to identify vertices that are local maxima of the labeling, which in this case means that they have label greater than or equal to their neighbours in the triangulation.
In the case of quadrangulations such property of a vertex could be established locally from the labeled tree just by considering the labels of tree edges incident to it.
Unfortunately this is not the case for triangulation, as can be seen from the examples in figure \ref{fig:trimax}.
Whether or not a vertex is a local maximum may depend on the structure in the tree an arbitrary distance away from that vertex.
Assigning couplings to local maxima is an inherently non-local procedure from the point of view of the labeled trees, and therefore quite impractical to treat analytically.

Instead, we will show that one can get a similar model by assigning couplings to saddle-point configurations in the triangulation.
In fact, we will introduce two couplings, $\g_1$ and $\g_2$, related to two different ways in which saddle points can occur. 
We assign a coupling $\g_1$ to each type-1 vertex labeled $t$ that is connected to more than one type-2 vertex labeled $t+1$.
In the triangulation this corresponds to a vertex labeled $t$ for which the labels of its neighbours run from $t$ to $t-1$ and back to $t$ at least twice when walking around the vertex.
We also assign a coupling $\g_2$ to each type-2 vertex labeled $t$ that is not connected to a type-1 vertex labeled $t-1$.
This corresponds to an edge labeled $(t,t)$ which is not shared by a \emph{down} triangle (see figure \ref{fig:trirules}).

Both these configurations are independent and correspond to configurations that are absent in causal triangulations.
If we set both couplings to zero and consider rooted, labeled trees for which all edges connecting to the root have constant label, we get exactly causal triangulations rooted at the ``top'', as in section \ref{sec:causal}.
Therefore it makes sense to search for generalized CDT in the continuum limit by scaling the couplings $\g_1$ and $\g_2$ to zero.

Let us introduce the generating functions $z_{i,\sigma}(g)$ for trees with a vertex of type $i=1,2$ at the root. 
The label $\sigma=-,0,+$ indicates whether the tree can be appear as a subtree of a larger tree with an edge of label $\sigma$ pointing towards the root of the subtree.
As can be deduced from the prescription in figure \ref{fig:trirules}, only four of these are occur, namely $z_{1,-}$, $z_{1,0}$, $z_{2,0}$, and $z_{2,+}$.
They satisfy the recurrence relations
\begin{align}
&z_{1,-} = \frac{1-\g_1}{1-g z_{2,0}} + \frac{\g_1}{1-g z_{2,+}-g z_{2,0}}, \\
&z_{1,0} =  \frac{1-\g_1}{1-g z_{2,0}}+g z_{2,+}\frac{1-\g_1}{(1-g z_{2,0})^2}+\frac{\g_1}{1-g z_{2,+}-g z_{2,0}}, \\
&z_{2,0} = g \left(z_{1,-} + \g_2(z_{1,0}+ z_{2,0}^2)\right), \label{eq:z2zero}\\
&z_{2,+} = g \left(z_{1,-} + z_{1,0} + z_{2,0}^2\right). \label{eq:z2plus}
\end{align}
If we scale $g=1/2(1-\lambda \epsilon^2/2)$ and $\g_i = \g_{s,i}\epsilon^3$ we find that the generating functions scale as $z_{1,-}=2(1-Z_{1,-}\epsilon)$, $z_{1,0} = Z_{1,0}/\epsilon$, $z_{2,0} = 1-Z_{2,0}\epsilon$, and $z_{2,+} = Z_{2,+}/\epsilon$, satisfying
\begin{align}
&Z_{1,0} = 2 Z_{2,+} = \frac{5}{2Z_{2,0}}, \quad Z_{1,-} = Z_{2,0}\\
&Z_{2,0}^3 - \lambda Z_{2,0} + \frac{5}{4} \g_{s,2} = 0 \label{eq:z2zerocont}.
\end{align}
Interestingly, only the coupling $\g_2$ survives in the continuum limit and produces the generalized CDT coupling $\g_s$ up to a factor of $5/4$.\footnote{This slightly awkward factor of $5/4$ can be seen to be due to the presence of flat triangles for small but non-zero $\g_{s,2}$. If one assigns yet another coupling $\g_3$ to each flat triangle, which amounts to inserting $\g_3$ in front of the $z_{2,0}^2$-terms in (\ref{eq:z2zero}) and (\ref{eq:z2plus}), and one scales $\g_3$ to zero at least linearly in $\epsilon$, the factor of $5/4$ will disappear and we get exactly $\g_{s,2}=\g_s$.}

The cup function $w(g,y)$ with a constant distance between the boundary and the origin is obtained simply by combining independent trees generated by $z_{2,0}$,
\begin{equation}
w(g,y) = \sum_{l=0}^{\infty} z_{2,0}(g)^l y^l = \frac{1}{1-y\, z_{2,0}(g)}.
\end{equation}
Taking the continuum limit of this expression, we arrive at the same result as for quadrangulations (\ref{eq:gencdtcup}) with $Z_1$ replaced by $Z_{2,0}$, which obeys the same equation (\ref{eq:z2zerocont}).

\section{Discussion and conclusions}\label{sec:discussion}

The model of two-dimensional quantum gravity denoted generalized 
causal dynamical triangulations was originally introduced as 
a continuum model of geometries. The equations which determined
the disk function and the two-loop function were found from
simple consistency relations which had to be satisfied
for the ensemble of geometries in question. Only afterwards an actual 
realization of this ensemble in terms of discrete triangulations was studied, 
from which a scaling limit could be found by taking the link length
to zero.
The purpose of this article has been to present another discrete realization of the model, in terms of labeled quadrangulations, which could be studied in more detail and be solved already largely at the discrete level.
Let us summarize some of the main results reported in this article:
\begin{itemize}
\item[1)] There exists a bijection $\Phi$ from the set of rooted pointed quadrangulations of the sphere with $N$ faces to the set of rooted pointed planar maps with $N$ edges such that the distance labeling of a quadrangulation $Q$ is mapped to the distance labeling of $\Phi(Q)$, and such that if $Q$ has $n$ local maxima with respect to the distance labeling then $\Phi(Q)$ has $n$ faces.    
 
\item[2)] An explicit generating function has been obtained for the number of rooted pointed quadrangulations of the sphere which have $N$ faces and $n$ local maxima with respect to the distance labeling and where the root is a distance $t$ from the origin.
This is also the generating function for the number of rooted pointed planar maps with $N$ edges and $n$ faces and where the root is a distance $t$ from the origin. 

\item[3)] We have shown that the generating functions for the ensemble of planar maps possess a well-defined scaling limit as $N\to\infty$, while keeping $n$ fixed.
In this limit one obtains precisely the amplitudes of the generalized CDT model of two-dimensional quantum gravity, which have been established in the physics literature.
Moreover, we have found a new explicit formula for the two-point function of generalized CDT.

\item[4)] We have shown that a loop identity first discovered by Kawai et al. in the continuum limit of DT, and seemingly related to the Virasoro algebra of the corresponding continuum quantum Liouville theory, is valid even at the combinatorial level of generalized CDT.
This loop identity is a consequence of the existence of a range of bijections between quadrangulations with two marked points and planar maps with two marked points.
\end{itemize}  

\begin{figure}[t]
\centering
\includegraphics[width=14cm]{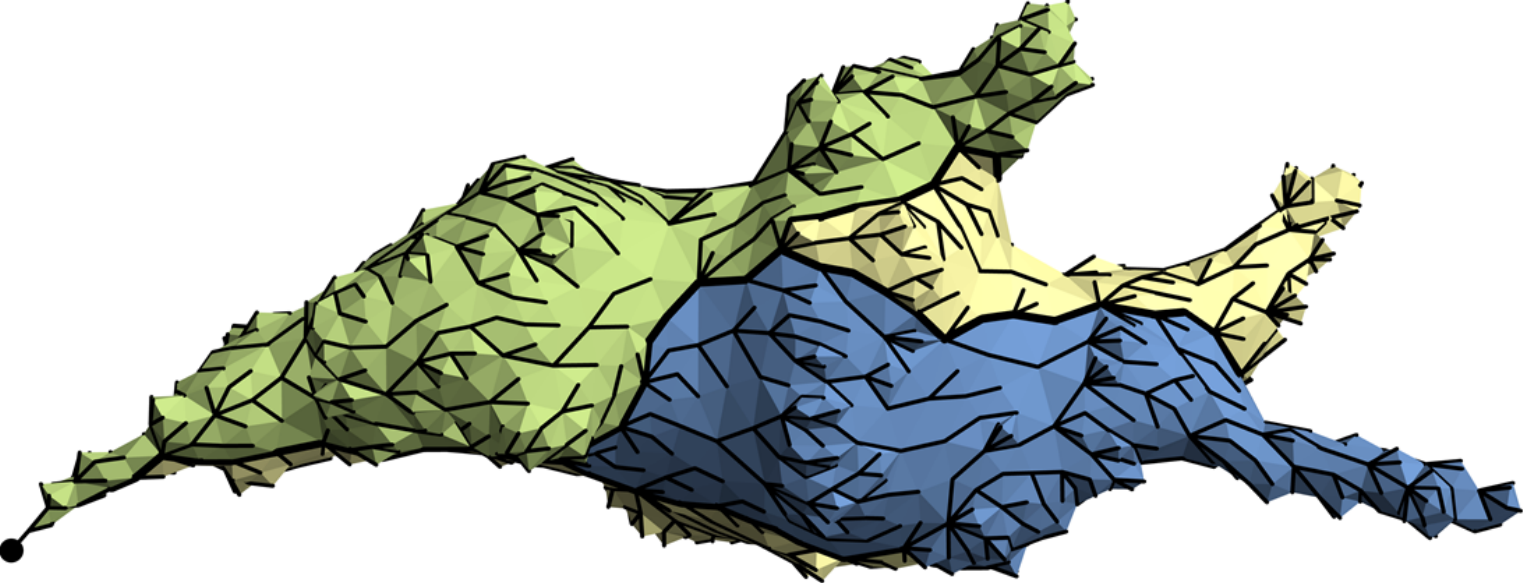}
\caption{A random pointed quadrangulation with the origin indicated by a black dot and three local maxima. The associated planar map drawn on top of the quadrangulation has three faces, as indicated by the colouring, of which the interiors may be viewed as baby universes. }%
\label{fig:gencdtmap}
\end{figure}

Let us make some remarks on the physical interpretation of generalized CDT as 2d gravity with spatial topology change.
In the two-dimensional CDT model ``space'' has the topology of $S^1$
and a two-dimensional CDT surface with the topology of a cylinder
allows a ``time'' foliation. In generalized CDT one allows the space 
with topology $S^1$ to split into several $S^1$'s as a function of time.
Each of these $S^1$'s develops independently and its splitting off can be viewed as the creation of a baby universe. As a function of time (or distance) 
we allow a finite number of these baby universes to be created
and eventually to vanish
again (``decay into the vacuum'' in physics jargon). 
Each spacetime point has a distance to the initial spatial $S^1$
and this distance has a local maximum where a baby universe vanishes.
In this way we were led to study
quadrangulations with a fixed number of local maxima of the distance 
with respect to the initial spatial $S^1$.
The outcome of our analysis is that this is combinatorially equivalent to 
the study of planar maps with a fixed number of faces 
and the generalized CDT model was obtained as a specific 
scaling limit of these planar maps.
In the planar map representation a baby universe is represented as a face and the 
spacetime volume of the baby universe is proportional to the
degree of the face (see figure \ref{fig:gencdtmap} for an example). 
In fact, inspection of the construction of the planar map shows that a face of degree $d$ covers exactly $d/2$ quadrangles.

This model may be compared with \cite{le_gall_scaling_2011} (see also \cite{janson_scaling_2012}), 
where random planar maps are studied with non-trivial weights on the 
degrees of the faces. By choosing different asymptotic laws for these weights, 
it was shown that different continuum limits are obtained 
with Hausdorff dimensions anywhere between 2 and 4.
It would be interesting to see whether putting non-trivial (hence non-local) weights 
on the volumes of baby universes leads to continuum limits 
that in a similar fashion continuously interpolate between DT and CDT, and if 
such weights can be given an interpretation in terms of 
continuum physics.

One should keep in mind, however, that the geometry studied in \cite{le_gall_scaling_2011} is the intrinsic geometry of the planar maps, which differs from the geometry of the associated quadrangulations.
Indeed, only geodesic distances to the origin are preserved under the bijection.
General methods to study distances between arbitrary points in generalized CDT, or even in ordinary CDT, are currently lacking.
This means that we have not yet arrived at a complete understanding of the two-dimensional continuum geometry of generalized CDT.

Finally, let us mention that recently it has been shown that one can obtain new multi-critical scaling relations in the generalized CDT ensemble of graphs if one 
combines triangles with quadrangles with negative
weights (\cite{ambjorn_new_2012} and \cite{atkin_analytical_2012,atkin_quantum_2012}). This is very similar 
to the now ``classical'' situation for the DT ensemble 
where a similar combination of weights for triangles and quadrangles
allowed one to obtain a  scaling limit different from the 
standard DT limit  as well as a different distance 
function (\cite{gubser_scaling_1994} and \cite{bouttier_geodesic_2003}). 
It should be possible to apply the techniques in \cite{bouttier_geodesic_2003}
to the generalized CDT ensemble and obtain the corresponding distance functions.  

\subsection*{Acknowledgments}

The authors acknowledge support from the ERC-Advance grant 291092,
``Exploring the Quantum Universe'' (EQU). JA acknowledges support
of FNU, the Free Danish Research Council, from the grant
``quantum gravity and the role of black holes''.        
  
  
\bibliographystyle{habbrv}
\bibliography{trees}

\end{document}